\documentclass[a4paper,twocolumn,superscriptaddress,11pt,accepted=2018-10-08]{quantumarticle}
\pdfoutput=1
\usepackage[utf8]{inputenc}
\usepackage[english]{babel}
\usepackage[T1]{fontenc}
\usepackage{amsmath}
\usepackage{hyperref}

\usepackage{tikz}
\usepackage{lipsum}

\usepackage{ulem} % only here for strikethrough-font
\usepackage{lineno}
\usepackage[caption=false]{subfig}
\usepackage{dcolumn}
\usepackage{amsmath,amssymb}
\usepackage{bm}
\usepackage{bbm}
\usepackage{overpic}
\usepackage{latexsym}
\usepackage{epstopdf}
\usepackage{color}
\usepackage[english]{babel}
\usepackage{latexsym}
\usepackage{psfrag,graphicx}
\usepackage{epsf}
\usepackage{amsmath}
\usepackage{amssymb}
\usepackage{amsfonts}
\usepackage[square,numbers]{natbib}
\usepackage{epstopdf}
\usepackage{appendix}
\usepackage{verbatim}
\usepackage{enumitem}
\usepackage{dsfont}
\usepackage{float}
\usepackage{tikz}
\makeatletter
\def\grd@save@target#1{%
\def\grd@target{#1}}
\def\grd@save@start#1{%
\def\grd@start{#1}}
\tikzset{
grid with coordinates/.style={
to path={%
  \pgfextra{%
    \edef\grd@@target{(\tikztotarget)}%
    \tikz@scan@one@point\grd@save@target\grd@@target\relax
    \edef\grd@@start{(\tikztostart)}%
    \tikz@scan@one@point\grd@save@start\grd@@start\relax
    \draw[minor help lines,magenta] (\tikztostart) grid (\tikztotarget);
    \draw[major help lines] (\tikztostart) grid (\tikztotarget);
    \grd@start
    \pgfmathsetmacro{\grd@xa}{\the\pgf@x/1cm}
    \pgfmathsetmacro{\grd@ya}{\the\pgf@y/1cm}
    \grd@target
    \pgfmathsetmacro{\grd@xb}{\the\pgf@x/1cm}
    \pgfmathsetmacro{\grd@yb}{\the\pgf@y/1cm}
    \pgfmathsetmacro{\grd@xc}{\grd@xa + \pgfkeysvalueof{/tikz/grid with coordinates/major step}}
    \pgfmathsetmacro{\grd@yc}{\grd@ya + \pgfkeysvalueof{/tikz/grid with coordinates/major step}}
    \foreach \x in {\grd@xa,\grd@xc,...,\grd@xb}
    \node[anchor=north] at (\x,\grd@ya) {\pgfmathprintnumber{\x}};
    \foreach \y in {\grd@ya,\grd@yc,...,\grd@yb}
    \node[anchor=east] at (\grd@xa,\y) {\pgfmathprintnumber{\y}};
  }
}
},
minor help lines/.style={
help lines,
step=\pgfkeysvalueof{/tikz/grid with coordinates/minor step}
},
major help lines/.style={
help lines,
line width=\pgfkeysvalueof{/tikz/grid with coordinates/major line width},
step=\pgfkeysvalueof{/tikz/grid with coordinates/major step}
},
grid with coordinates/.cd,
minor step/.initial=.2,
major step/.initial=1,
major line width/.initial=2pt,
}
\makeatother

\makeatletter
\def\resetMathstrut@{%
\setbox\z@\hbox{%
    \mathchardef\@tempa\mathcode`\[\relax
    \def\@tempb##1"##2##3{\the\textfont"##3\char"}%
    \expandafter\@tempb\meaning\@tempa \relax
}%
\ht\Mathstrutbox@\ht\z@ \dp\Mathstrutbox@\dp\z@}
\makeatother
\begingroup
\catcode`(\active \xdef({\left\string(}
\catcode`)\active \xdef){\right\string)}
\endgroup
\mathcode`(="8000 \mathcode`)="8000

\definecolor{mygrey}{gray}{0.35}
\definecolor{myblue}{rgb}{0.2,0.2,0.8}
\definecolor{myzard}{cmyk}{0,0,0.05,0}
\definecolor{mywhite}{rgb}{1,1,1}
\definecolor{myred}{rgb}{0.9,0.1,0.}
\newtheorem{theorem}{Theorem}
\newtheorem{lemma}[theorem]{Lemma}

\newtheorem{definition}[theorem]{Definition}
\newtheorem{corollary}[theorem]{Corollary}

\newtheorem{proposition}[theorem]{Proposition}

\newenvironment{proof}{\medskip\noindent\textbf{Proof.}}{\hfill$\blacksquare$\medskip}
\newenvironment{proof-of}[1]{\medskip\noindent\textbf{Proof of {#1}.}}{\hfill$\blacksquare$\medskip}

\newcommand{\tr}{\operatorname{\bf{tr}}} % trace of a matrix

\newcommand{\simu}{\operatorname{sim}} % sim subindex
\newcommand{\ket}[1]{\vert #1 \rangle} % ket vector
\newcommand{\bra}[1]{\langle #1 \vert} % bra vector
\newcommand{\id}{\operatorname{id}}
\newcommand{\1}{\mathds{1}}

\newcommand{\cM}{{\cal M}}
\newcommand{\cN}{{\cal N}}
 \newcommand{\cE}{{\cal E}}
 \newcommand{\cL}{{\cal L}}
 \newcommand{\cT}{{\cal T}}

\newcommand{\simcoz}[1]{C_{\mathrm{sim}}\left(#1\right)}
\newcommand{\simco}[1]{C_{\mathrm{sim}}^\epsilon\left(#1\right)}
\newcommand{\amocoz}[1]{C_{\mathrm{amo}}\left(#1\right)}
\newcommand{\amoco}[1]{C_{\mathrm{amo}}^\epsilon\left(#1\right)}
\newcommand{\bfrob}[1]{\ensuremath{1+C_R(#1)}}
\newcommand{\eqq}[1]{Eq.~(\ref{#1})}

\newcommand{\dketbra}[1]{\ensuremath{| #1 \rangle\!\langle #1 |}}

\newcommand{\braket}[2]{\ensuremath{\langle #1 | #2 \rangle}}
\newcommand{\ketbra}[2]{\ensuremath{| #1 \rangle\!\langle #2 |}}
\newcommand{\ms}[1]{{\color{black}#1}}
\newcommand{\jcc}[1]{{\color{Red}#1}} % you know because he is catalan :)
\begin{document}

\title{Using and reusing coherence to realize quantum processes}

\author{Mar\'ia Garc\'ia D\'iaz}
\affiliation{F\'{\i}sica Te\`{o}rica: Informaci\'{o} i Fen\`{o}mens Qu\`{a}ntics, %
Departament de F\'{\i}sica, Universitat Aut\`{o}noma de Barcelona, ES-08193 Bellaterra (Barcelona), Spain}

\author{Kun Fang}
\affiliation{Centre for Quantum Software and Information, University of Technology Sydney, NSW 2007, Australia}

\author{Xin Wang}
\affiliation{Centre for Quantum Software and Information, University of Technology Sydney, NSW 2007, Australia}

\author{Matteo Rosati}
\affiliation{F\'{\i}sica Te\`{o}rica: Informaci\'{o} i Fen\`{o}mens Qu\`{a}ntics, %
Departament de F\'{\i}sica, Universitat Aut\`{o}noma de Barcelona, ES-08193 Bellaterra (Barcelona), Spain}

\author{Michalis Skotiniotis}
\affiliation{F\'{\i}sica Te\`{o}rica: Informaci\'{o} i Fen\`{o}mens Qu\`{a}ntics, %
Departament de F\'{\i}sica, Universitat Aut\`{o}noma de Barcelona, ES-08193 Bellaterra (Barcelona), Spain}

\author{John Calsamiglia}
\affiliation{F\'{\i}sica Te\`{o}rica: Informaci\'{o} i Fen\`{o}mens Qu\`{a}ntics, %
Departament de F\'{\i}sica, Universitat Aut\`{o}noma de Barcelona, ES-08193 Bellaterra (Barcelona), Spain}

\author{Andreas Winter}
\affiliation{F\'{\i}sica Te\`{o}rica: Informaci\'{o} i Fen\`{o}mens Qu\`{a}ntics, %
Departament de F\'{\i}sica, Universitat Aut\`{o}noma de Barcelona, ES-08193 Bellaterra (Barcelona), Spain}
\affiliation{ICREA---Instituci\'o Catalana de Recerca i Estudis Avan\c{c}ats, %
Pg.~Lluis Companys, 23, ES-08001 Barcelona, Spain}

\begin{abstract}
Coherent superposition is a key feature of quantum mechanics that underlies the advantage of quantum technologies over their classical counterparts. 
Recently, coherence has been recast as a resource theory in an attempt to identify and quantify it in an operationally well-defined manner. Here we study how the coherence present in a state can be used to implement a quantum channel via incoherent operations and, in turn, to assess its degree of coherence. 
We introduce the robustness of coherence of a quantum channel---which reduces to the homonymous measure for states when computed on constant-output channels---and prove that: i) it quantifies the minimal rank of a maximally coherent state required to implement the channel; ii) its logarithm quantifies the amortized cost of implementing the channel provided some coherence is recovered at the output; iii) its logarithm also quantifies the zero-error asymptotic cost of implementation of many independent copies of a channel.
We also consider the generalized problem of imperfect implementation with arbitrary resource states. Using the robustness of coherence, we find that in general a quantum channel can be implemented without employing a maximally coherent resource state. 
In fact, we prove that \textit{every} pure coherent state in dimension larger than $2$, however weakly so, turns out to be a valuable resource to implement \textit{some} coherent unitary channel.
We illustrate our findings for the case of single-qubit unitary channels.
\end{abstract}

\maketitle

\ms{Quantum theory has profoundly altered the way we view the physical world.  
Since its concrete formulation roughly a century ago the theory has 
produced startling predictions, such as quantum entanglement~\cite{einstein1934} 
and the no-cloning theorem~\cite{wootters1982} to name just two, 
that bear no analog in our classical theories of physics and 
information processing.  
The underlying cause of departure of quantum theory from classical 
ways of thought---and perhaps its only mystery~\cite{feynman1965}---is the 
principle of superposition. Indeed, the fact that quantum systems can 
exist in a superposition state underlies every advantage furnished by 
quantum technologies over their classical counterparts, from 
thermodynamics~\cite{lostaglio2015,lostaglio2017,faist2015a,misra2016,korzekwa2016} and non-classicality~\cite{killoran16},  condensed matter~\cite{cakmak2015,malvezzi2016}, 
metrology~\cite{giovanetti2011,toth2014, degen2017} and 
atomic clocks~\cite{ludlow2015}, 
quantum simulation~\cite{buluta2009,georgescu2014}
and computation~\cite{galindo2002,ladd2010,hillery2016,anand2016,matera2016} 
to communication~\cite{gisin2002, gisin2007,holevo2012}, and has even 
been employed in the description of fundamental biological 
processes~\cite{engel2007,ishizaki2009,chin2013,huelga2013}.

The aforementioned advantages render quantum superposition a precious 
resource; having access to it allows one to perform tasks 
that are otherwise impossible. Hitherto, the focus has been primarily on the quantification and interconversion of \emph{static coherence}---the degree of superposition present in a state~\cite{baumgratz2014,winter2016,aberg2006}.  
This overlooks the principal goal of a resource, namely that it can be spent in order to perform non-trivial tasks.  Some initial steps in this direction have been undertaken for some particular tasks~\cite{baumgratz2014,bendana2017,napoli2016,biswas2017,giorda2017,hillery2016}. Here we establish a general framework for converting static coherence into \emph{dynamic coherence}---the ability to generate coherence itself. 
Specifically, we quantify the amount of static coherence used to implement an arbitrary quantum process and show that in general some static coherence can be reused after the implementation. We prove that both these quantities are quantified by the same measure that we call channel robustness of coherence, in analogy to that of states.  

The coherence of a state can be cast within the 
general framework of a quantum resource theory~\cite{brandao2015a}.
Building upon the works of Baumgratz {\it et al.}~\cite{baumgratz2014} 
and \AA{}berg~\cite{aberg2006} (see also~\cite{BraunGeorgeot}), the 
construction of a full-fledged resource theory of coherence, 
much like those of entanglement~\cite{plenio2007,vedral1998,horodecki2009,eisert2000}, 
asymmetry~\cite{gour2008,marvian2014}, and 
athermality~\cite{brandao2013,brandao2015b,horodecki2013,faist2015b,gour2015,narasimhachar2015}, 
has received increased attention in recent years.
\begin{figure}[t!]
\subfloat[\label{fig:settingnoancilla}]{%
	\includegraphics[width=.5\linewidth]{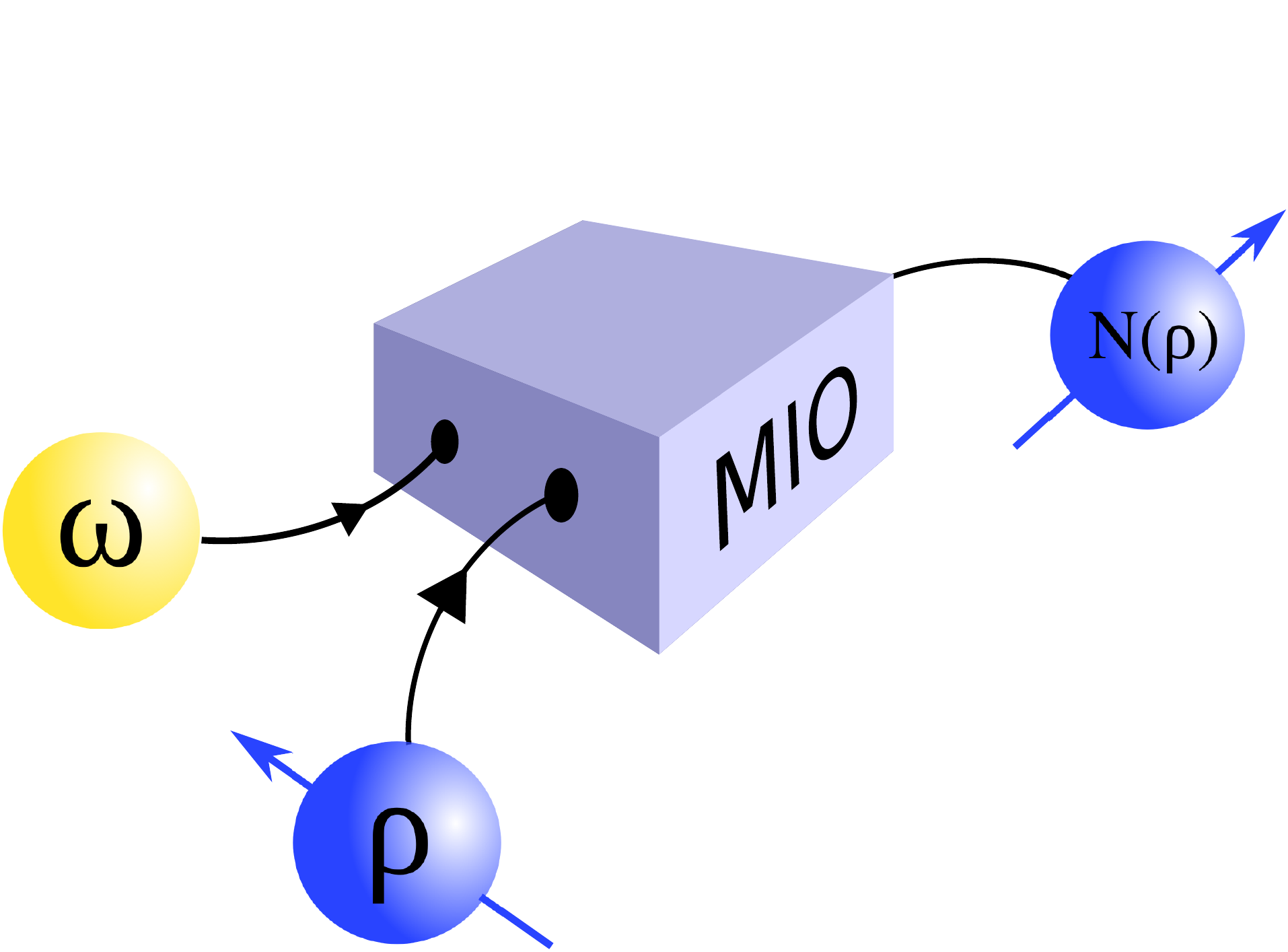}%
}
\subfloat[\label{fig:setting}]{%
	\includegraphics[width=.5\linewidth]{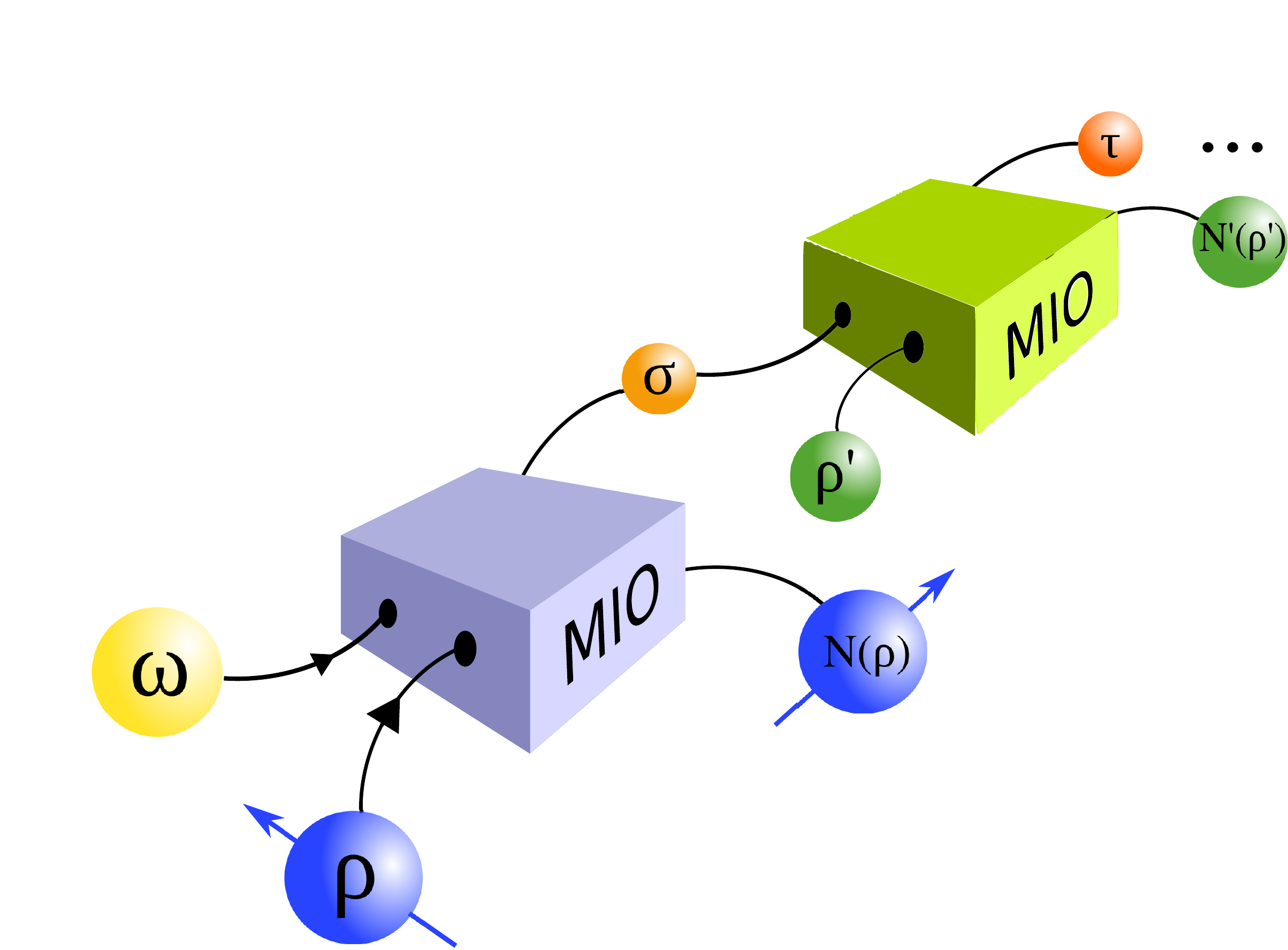}%
}
\caption{{Different protocols for implementing a quantum channel $\cN$ using a MIO  
		$\cM$ and a coherent resource $\omega$. 
		\textbf{a)} The implementation destroys completely the input resource,
		$\cM(\omega \otimes \rho)=\cN(\rho)$. 
		\textbf{b)} After the implementation, a degraded resource $\sigma$ is 
		recovered and ready to be recycled: $\cM(\omega\otimes\rho)=\sigma \otimes \cN(\rho)$. 
		For example it can be used to implement another channel $\cN'$:
		$\cM'(\sigma\otimes\rho') = \tau \otimes \cN'(\rho')$.
}}
\end{figure}
In the general framework of quantum resource theories one first identifies 
which states and operations are free and then studies how resource states 
can be quantified, manipulated, and interconverted among each other. 
In the resource theory of coherence 
the set $\Delta$ of free states---called \emph{incoherent states}---consists 
of all the states $\delta\in\mathcal{S(H)}$ that are diagonal in 
some fixed basis $\{\ket{i}\}_{i=0}^{d-1}$ of $\mathcal{H}$, 
i.e.~\mbox{$\Delta=\left\{\sum_i p_i\dketbra{i}\,:\,\sum_i p_i=1\right\}$}, 
determined by the problem under investigation.   
Free operations are those that map the set of free states to itself, 
hence do not generate coherence.
Unlike the resource theories of entanglement, asymmetry and athermality, 
where the set of free operations is uniquely identified on operational 
grounds, no such consensus exists within the resource theory of 
coherence~\cite{chitambar2016a}.}
The largest class of free operations are the 
\emph{maximally incoherent operations} (MIOs)~\cite{aberg2006}, 
comprising all completely positive and trace-preserving (CPTP) quantum 
channels $\cM$ such that $\cM(\Delta)\subset \Delta$. A subset of MIOs 
are \emph{incoherent operations} (IOs)~\cite{baumgratz2014} which 
consist of all quantum channels $\cM$ admitting a Kraus representation 
with operators $K_\alpha$ such that $K_\alpha \Delta K_\alpha ^\dagger \subset \Delta$ 
for all $\alpha$. 
\ms{Finally, a \emph{coherence measure}~\cite{baumgratz2014} is a functional 
$C:\mathcal{S(H)}\to\mathbb{R}_{\geq0}$, quantifying the amount 
of coherence present in a state $\rho\in\mathcal{S(H)}$ that satisfies
the following two conditions; (i) faithfulness, which means that $C(\delta)=0$ for all incoherent $\delta\in\Delta$, 
and (ii) monotonicity, $C(\rho)\geq C(\cM(\rho))$, for all free operations $\cM$.  
Examples of coherence measures include the relative entropy of 
coherence~\cite{baumgratz2014}, the robustness of coherence~\cite{napoli2016}, the coherence rank (coherence number) \cite{killoran16}} 
and the $\ell_1$-norm of coherence~\cite{baumgratz2014}; the latter two 
are valid measures only under IOs as free operations, but not MIOs (see Appendix~\ref{app:RoC} 
for further details, and~\cite{streltsov2017} for an extensive review).  
Finally, the \emph{maximally coherent state} in dimension $d$, irrespectively of the coherence measure of choice, is the state $\ket{\Psi_d}=\frac{1}{\sqrt{d}}\sum_{j=0}^{d-1}\ket{j}$, 
which we will henceforth call \emph{cosdit}, reserving the term 
\emph{cosbit} for the unit of quantum coherence---the maximally coherent qubit.

In this work, we go beyond the resource theory of states and investigate how the coherence present in a state can be used to implement a quantum channel. We thereby assess the coherence of the channel itself, in the spirit of channel resource theory~\cite{bendana2017,BraunGeorgeot}. Previous approaches addressed only the cases where the available operations are IOs and asymptotically many rounds of implementation are allowed~\cite{baumgratz2014,chitambar2016b,bendana2017,BraunGeorgeot}. Here we benchmark the ultimate performance of implementation by using MIOs as free operations. Moreover, we mainly consider single-shot channel implementation and discuss its relevance on asymptotics in Sec.~\ref{distillation}.
We find necessary and sufficient conditions for a channel to 
be implemented with cosdit resource states of any chosen dimension (Fig.~\ref{fig:settingnoancilla}), significantly advance the results of previous works on the same subject~~\cite{baumgratz2014,chitambar2016b,bendana2017}. 
The minimal dimension of such a cosdit is used to quantify the \emph{simulation cost} of the channel.
In addition, we consider the setting where the input cosdit is 
recovered in a degraded form at the output, after implementing the required 
channel. This setting is particularly appealing since the initial resource 
can be recycled for further use (Fig.~\ref{fig:setting}) and it naturally 
gives rise to the notion of an \emph{amortized simulation cost} as the 
difference between the input and output coherence, 
which quantifies the resources that are actually consumed by the implementation. 

We begin by introducing the robustness of coherence of quantum channels in 
Sec.~\ref{cosdits} and single it out as a natural coherence measure for channels, as it 
quantifies the simulation and amortized costs of a channel with cosdit resource states.
In Sec.~\ref{noncosdits} we consider instead the implementation of channels with arbitrary resource states, providing a semidefinite program to compute the minimum 
implementation error and the amortized cost.
We then introduce a family of non-maximally coherent resource states 
that allow for the implementation of arbitrary quantum channels and prove 
that in fact \emph{every} pure coherent state of dimension $d>2$ is useful for the 
\emph{exact} implementation of \emph{some} coherent unitary channel. 
Finally, we focus on the case of 
qubit unitaries: (i) we prove that, among qubit resource states, a cosbit 
is necessary and sufficient to implement any coherent qubit unitary; (ii) for the case of qutrit resource states we provide a full characterization with the aid of numerics; (iii) we show that MIOs, unlike IOs~\cite{bendana2017}, allow for coherence recycling.
Eventually, in Sec.~\ref{distillation} we discuss the relation of our findings to the problem of asymptotic reversibility between coherence cost and coherence generating capacity of quantum channels.
For ease of exposition, we defer the proofs of all Theorems, Lemmas and 
Propositions to Appendices.

\section{Results}
\subsection{MIO implementation of quantum channels\protect\\ using maximally coherent resources}
\label{cosdits}

The main goal of the present article is to quantify the resources required to
implement (or simulate) an arbitrary quantum channel via MIOs by making use of coherent input states. There are many inequivalent ways to quantify the coherence of a given state. Indeed, a state can be more resourceful than another according to a given measure, while the opposite can happen according to a different one. This is due to the existence of \textit{incomparable} resources, i.e., pairs of states that cannot be interconverted in either direction via MIOs\footnote{For instance, the cosbit cannot be transformed via MIO into the flagpole state $\ket{\varphi_{\frac23}}$ (defined below in Eq. \eqref{eq:flagpole}), since the latter has a greater robustness of coherence $C_{R}(\varphi)=\frac53>1=C_{R}(\Psi_{2})$; and at the same time the flagpole $\ket{\varphi_{\frac23}}$ cannot be transformed via MIO into the cosbit by Lemma~\ref{lemma1}.}. Nevertheless, in this section we wish to establish general results which hold irrespectively of the particular coherence measure of choice: we will do this at the expense of quantifying the required resource in a somewhat coarse fashion, namely by the smallest size $k$ of a cosdit $\ket{\Psi_k}$ that allows to implement the channel.
Indeed, cosdits are maximally coherent states irrespectively of how coherence is quantified, since any state of the same (or lower) coherence rank can be obtained from them via MIOs~\cite{baumgratz2014}. Restricting our input resources to cosdits might seem a limitation. Nevertheless, we will show that this setting has clear benefits and that it leads to a coherence measure for channels that does take real values.

Our starting point is to introduce the channel robustness: 
\begin{definition}\label{RobDef}
The \emph{robustness of coherence} of a quantum channel $\cN$, $C_R(\cN)$, is
defined by
\begin{equation}
  \label{robCohChDef}
  1+C_R(\cN) := \min \left\{\lambda: \  \cN \leq \lambda \cM,\ \cM \in\text{MIO}\right\},
\end{equation}
where the inequality $\cN\leq\lambda \cM$ is understood as completely-positive 
ordering of super operators, i.e., $\lambda \cM - \cN$ is a CP map. 
The smoothed version of this quantity, 
called \emph{$\epsilon$-robustness of coherence}, is given by
\begin{equation}
  \label{smoothRob}
  C_R^\epsilon(\cN)=\min \left\{C_R(\cL): \ \frac12 \|\cN-\cL\|_\diamond\leq \epsilon\right\},
\end{equation}
where $\|\cdot \|_\diamond$ denotes the diamond norm~\cite{aharonov1998,watrous2009}. 
\end{definition}
Recall that the diamond norm is a well-behaved distance which is furthermore endowed with an operational meaning: it quantifies how well one can physically discriminate between two quantum channels \cite{sacchi2005}.
Also note that definition \ref{RobDef} reduces to the robustness of coherence of states~\cite{napoli2016}
when applied to the constant channel $\cN_\sigma(\rho)=\sigma$, i.e.,
$C_R(\cN_\sigma)=C_R(\sigma)$.
It is straightforward to see (Appendix~\ref{app:RobCoh}) that, 
just as in the case of states, the robustness of coherence of a 
channel $\cN$ quantifies the minimum amount of noise, in the form of 
another channel, that we need to add to $\cN$ such that the resulting 
channel is MIO. In Appendix~\ref{app:RobCoh} we also show that this 
quantity is a proper coherence measure for channels under MIO and 
that it can be easily formulated as a semidefinite program. Such 
formulation allows for an efficient numerical computation of this 
measure and, more importantly, thanks to its dual form, it facilitates 
the proof of the following theorem. In it, $P_{R}(\cN)$ \cite{mani2015,bu2017a} and $\widehat{P}_{R}(\cN)$ \cite{garcia2016, bu2017a} are two alternative, seemingly inequivalent definitions of the cohering power:
\begin{align}
P_R(\cN) &= \max_{\rho\in \Delta} C_{LR}(\cN(\rho))\nonumber \\
&= \max_i \log(1+C_R(\cN(\ket{i}\!\bra{i}))), \label{eq:CRNi} \\
\widehat P_R(\cN) &= \max_{\rho\in \mathcal{S(H)}}(C_{LR}(\cN(\rho))-C_{LR}(\rho)),
\end{align}
with $C_{LR}(\rho) = \log(1+C_R(\rho))$ being the log-robustness of a 
quantum state $\rho$~\cite{napoli2016} (see  Appendix~\ref{app:RoC}); 
in particular, for a cosdit it holds $C_{LR}(\Psi_k)=\log k$. 

\begin{theorem}
  \label{theorem_dual}
  The log-robustness of coherence of a channel $\cN$,
  defined in the following equation, can be expressed as
  \begin{equation}\label{robCohPow}
    C_{LR}(\cN):=\log (1+C_R(\cN))=P_{R}(\cN)=\widehat{P}_{R}(\cN)
  \end{equation}
  and it is additive under tensor product of channels. 
  %, i.e., $1+C_R(\cN_1\otimes \cN_2)=\left(1+C_R(\cN_1)\right)\left(1+C_R(\cN_2)\right)$.
\end{theorem}
Note that throughout the article the base-2 $\log$ is employed. Moreover, it can be shown that the (smooth) log-robustness of a channel $\cN$ equals the (smooth) maximum relative entropy between $\cN$ and a MIO $\cM$, minimized over all $\cM$, that we define in Appendix~\ref{app:RobCoh}.
%In Theorem \ref{theorem_dual}, $P_{R}(\cN)$ \cite{mani2015,bu2017a} and $\widehat{P}_{R}(\cN)$ \cite{garcia2016, bu2017a} are two alternative, seemingly inequivalent definitions of the cohering power:
%\begin{align}
%	P_R(\cN) &= \max_{\rho\in \Delta} C_{LR}(\cN(\rho))\nonumber \\
%	            &= \max_i \log(1+C_R(\cN(\ket{i}\!\bra{i}))), \label{eq:CRNi} \\
%	\widehat P_R(\cN) &= \max_{\rho\in \mathcal{S(H)}}(C_{LR}(\cN(\rho))-C_{LR}(\rho)),
%\end{align}
%with $C_{LR}(\rho) = \log(1+C_R(\rho))$ being the log-robustness of a 
%quantum state $\rho$~\cite{napoli2016} (see  Appendix~\ref{app:RoC}); 
%in particular, for a cosdit it holds $C_{LR}(\Psi_k)=\log k$. 
We then see that Theorem~\ref{theorem_dual} states that the log-robustness of a channel has an operational meaning as its cohering power.
Moreover, the robustness and log-robustness play a prominent role in quantifying the cost and consumption of coherence in the implementation of a channel via MIOs, respectively, as we now explain.

Let us analyze the implementation setting of Fig.~\ref{fig:settingnoancilla}, when cosdits are employed as resources:
\begin{definition}
	A MIO simulation of a quantum channel $\cN:A\rightarrow B$ up to error $\epsilon$ with a cosdit $\Psi_k$ is a MIO $\cM:R\otimes A\rightarrow B$ that satisfies
	\begin{equation}\label{close}
		\dfrac{1}{2}\|\cN-\cM(\Psi_k\otimes \cdot)\|_\diamond\leq \epsilon.
	\end{equation}
The simulation cost of $\cN$, denoted $\simco{\cN}$, is the smallest $\log k$ satisfying \eqq{close}, i.e., the minimal coherence rank of the resource allowing for a MIO simulation of $\cN$.
\end{definition}

%Backup
%\begin{definition}
%	An $(\epsilon,k)$-MIO simulation of a quantum channel $\cN:A\rightarrow B$ with error $\epsilon$ and a maximally coherent resource of coherence rank $k$, i.e., $\Psi_k$, %is a MIO $M:R\otimes A\rightarrow B$ that satisfies
%	\begin{equation}
%		\dfrac{1}{2}||\cN-\cM(\Psi_k\otimes \cdot)||_\diamond\leq \epsilon.
%	\end{equation}
%\end{definition}
%Our purpose in this setting is to determine the minimum coherence rank of the cosdit resource that allows for the MIO implementation of a given channel up to some error.  %\ms{The cost of the simulation is defined to be:} 
%i.e., to calculate the simulation cost of the channel.
%\begin{definition}
		%The simulation cost of a quantum channel $\cN$ up to error $\epsilon$, denoted $C_{\simu}^{\mio}(\cN,\epsilon)$, is the smallest $\log k$ such that there exists an $%(\epsilon,k)$-MIO simulation of $\cN$.
%\end{definition}

{Henceforth, $C_{\mathrm{sim}}(\cN)$ implies $\epsilon=0$, i.e., 
exact implementation. We are now equipped to address one of the 
main objectives of this article and give an exact expression for 
the simulation cost of a channel in terms of its smoothed robustness of coherence: }
\begin{theorem}
  \label{theorem1}
  For any quantum channel $\cN$ it holds
  \begin{equation}
    \label{eq_MIO-simulation}
    \simco{\cN}= \log\lceil1+C_R^\epsilon(\cN)\rceil,
  \end{equation}
  where $\lceil\cdot\rceil$ is the ceiling function.
\end{theorem}
Theorem~\ref{theorem1} can be seen as one-shot coherence 
dilution in the resource theory of channels: a maximally coherent resource 
is completely consumed to generate a target one with smaller 
coherence-generating capability. Indeed, this includes as a special 
case the one-shot coherence dilution of states recently studied 
in~\cite{zhao2018} and generalizes the criterion found in \cite{chitambar2016a} for transformations of cosdits $\Psi_k$ to pure states $\ket{\phi}=\sum_i\sqrt{p_i}\ket{i}$:
 \begin{equation}
    \label{planeconv}
    \sum_{i}\sqrt{p_{i}}\leq \sqrt{k}.
  \end{equation}

Thanks to Theorem \ref{theorem1}, 
we can now determine the minimal coherence rank of a cosdit necessary 
for a MIO implementation of a channel\footnote{We also note that, since the the cosdit is the maximally coherent state with a given coherence rank, Theorem~\ref{theorem1} implies that a generic input resource state $\omega$ will not be able to $\epsilon$-simulate a channel if its coherence, as measured by the coherence rank (or by coherence number if the resource is a mixed state, both defined in Appendix~\ref{app:RoC}), is smaller than the simulation cost: $C_{\mathrm{rank}}(\omega)< \lceil1+C_R^\epsilon(N)\rceil$. Indeed, the latter inequality implies that a cosdit $\psi_k'$ of the same rank $k'=C_{\mathrm{rank}}(\omega)$ cannot simulate the channel; moreover, by construction $\omega$ is less coherent that $\psi_k'$, since it can be obtained from the cosdit via $IO\subset MIO$~\cite{baumgratz2014}; hence $\omega$ cannot simulate a channel already not simulable with $\psi_k'$.}.
This provides a somewhat coarse measure of the 
implementation cost, as cosdits come in discretized units, the smallest
being $k= 2$. In order to quantify the actual resource consumed in 
the process, in the remaining part of this subsection we investigate 
to which extent some of the coherence of the input resource can be 
recovered after the channel implementation.
For this purpose let us now focus on the setting of Fig.~\ref{fig:setting}, 
where the resource, initially in a state $\omega$, 
is recovered in a degraded form $\sigma$ after fueling 
the MIO implementation of the target channel $\cN(\rho)$. We quantify the 
minimal resource consumed in this process by the difference between the coherence of $\omega$ and $\sigma$, when both the input and output states come in standard coherence units, i.e., $\omega$ and $\sigma$ are cosdits.
\begin{definition}\label{amorDef}
	The $\epsilon$-error amortized cost of a quantum channel $\cN$ is given by
	\begin{equation}
	  \label{eq_amortized}
	  \begin{split}
      \amoco{\cN} &:= \inf \ (C_{LR}(\Psi_k)-C_{LR}(\Psi_m))\\
      &=\inf (\log \frac{k}{m})\\
      \text{ \rm s.t. } &\ \cM\in MIO,\\
      &\ \cM(\Psi_k \otimes \cdot) = \Psi_m \otimes \cL(\cdot),\\
      &\ \frac{1}{2}||\cN-\cL||_\diamond\leq\epsilon,
      \end{split}
      \end{equation}
      where we recall that $C_{LR}(\Psi_k)=\log k$.
In other words the optimization is over all channels $\cL$ that are $\epsilon$-close to the target channel $\cN$ and that can be implemented via a MIO $\cM: R \otimes A \rightarrow  S \otimes B$ with cosdit input and output resource states, respectively $\Psi_k$ and $\Psi_m$, $k\geq m$.
\end{definition}
Note that tensor-product structure at the output of the simulation allows complete freedom in reusing $\sigma$ and $\mathcal{N}(\rho)$ afterwards; an entangled output would unnaturally constrain the recycling operations. For example, it would not allow the implementation of a sequence of channels applied on-the-fly to the same system, i.e., $\cN_n\circ\cdots\circ\cN_2\circ\cN_1$, whereas this is allowed by \eqq{eq_amortized}, see Corollary~\ref{corollary3} below.
 
Interestingly, the amortized cost defined in Eq.~\eqref{eq_amortized} can be related to the log-robustness of a channel. In order to show it, let us first discuss the exact implementation of a channel via MIOs with coherence recycling, which amounts to taking $\epsilon=0$ above and $\cL\equiv\cN$. In this case, thanks to Theorem~\ref{theorem1}, it is possible to estimate the robustness of coherence left in the resource after the implementation:
\begin{theorem}
  \label{theoremk-sim}
  For a quantum channel $\cN:A\rightarrow B$, there exists a MIO
  $\cM: R \otimes A \rightarrow  S \otimes B$ such that 
  $\cM(\Psi_k \otimes \cdot) = \sigma \otimes \cN(\cdot)$ if and only if
  \begin{equation}
    \label{eq_coh_left}
    C_{LR}(\cN) \leq C_{LR}(\Psi_{k}) - C_{LR}(\sigma).
  \end{equation}
\end{theorem}
Note that the bound in Eq.~\eqref{eq_coh_left} can always be attained by some resource state $\sigma_0$ that, however, will not be a maximally coherent. If we do impose that the output resource is a cosdit $\Psi_m$, we obtain the following result:
\begin{corollary}
  \label{corollary2}
  Given a quantum channel $\cN$ and an integer $k \geq 1+C_R(\cN)$,
  there always exists a MIO implementation of $\cN$ that takes a cosdit resource of coherence rank 
  $k$ and returns a degraded resource in the form of a cosdit 
  of rank \mbox{$m = \left\lfloor \frac{k}{1+C_R(\cN)} \right\rfloor$}, i.e.
\begin{equation}\label{interval}
    1+C_R(\cN) \leq \frac{k}{m} \leq (1+C_R(\cN))(1+ \frac{1}{m}).
\end{equation}
 % In particular this implies
%\begin{equation}
   %\amocoz{\cN} = C_{LR}(\cN) = \inf_{k,m\leq k} (\log k-\log m),
%\end{equation}
  %where the infimum is over all MIOs $\cM$ such that $\cM(\Psi\otimes\cdot) = \Psi_m\otimes \cN$.
\end{corollary}
That is, demanding a cosdit at the output requires an overhead of at 
most $O(1/m)$ with respect to the optimal ratio that can be attained 
with a non-maximally coherent output resource; moreover, this overhead can be made arbitrarily small by simply providing a higher-rank cosdit resource at the input. This implies that the amortized cost of a channel is equal to its log-robustness and the same straightforwardly extends to the approximate case, as proved by the following theorem:
\begin{theorem}  
  \label{amortizedc}
   For any quantum channel $\cN$ it holds
  \begin{equation}    
  \label{eq:amor}
               \amoco{\cN}   = C_{LR}^\epsilon(\cN),
\end{equation}
where
   \begin{equation}
C_{LR}^\epsilon(\cN)=\min \left\{C_{LR}(\cL): \ \frac12 \|\cN-\cL\|_\diamond\leq \epsilon\right\}.
\end{equation}          
\end{theorem}
This second key result, together with Theorem~\ref{theorem1}, establishes the robustness of coherence of a channel as the correct measure to quantify the cosdit resources necessary to implement the channel using a single MIO. Eq.~\eqref{eq_MIO-simulation} determines the minimum coherence rank required at the input, while Eq.~\eqref{eq:amor} determines the minimum fraction of input coherence that is actually used in the process. Note however that the latter is a lower bound on the actual coherence consumed when a limited amount of resource can be employed at the input.

Finally, restricting again to the zero-error case, Corollary~\ref{corollary2} paves the way for the exact MIO implementation of arbitrary sequences of channels, as depicted in Fig.~\ref{fig:setting}.
\begin{corollary}
  \label{corollary3}
  Any succession of channels $\cN_1, \cN_2, \cdots, \cN_n$ can be implemented on-the-fly by concatenation of MIO implementations 
  at an asymptotically optimal amortized cost 
  \begin{equation}
  \sum_{i=1}^{n} \amocoz{\cN_{i}}\leq \log\frac{k}{k'} \leq \sum_{i=1}^{n} \amocoz{\cN_{i}} + \frac{n}{k'},
  \end{equation}
  where $k$ and $k'\leq k$ are the coherence ranks of the input and 
  output cosdits of the entire protocol.
\end{corollary}
Note that the additivity of the zero-error log-robustness under tensor product, Theorem~\ref{theorem_dual}, already implies $\amocoz{\cN_1\otimes\cdots\otimes\cN_n}=\sum_{i=1}^n\amocoz{\cN_i}$. However, Corollary~\ref{corollary3} is more general and includes this result as a special case. Indeed, it allows the exact implementation of \emph{arbitrary sequences} of channels, and not just their tensor product, at an amortized cost equal to the sum of the single amortized costs of each channel. In particular, this allows the implementation of a concatenation of channels on-the-fly and it justifies the choice of a tensor-product structure at the output of the recycling process.
Finally, from Theorem \ref{theorem_dual} we know that the log-robustness of coherence also quantifies the cohering power of a quantum channel. Hence we conclude that the exact amortized cost of a channel coincides with its cohering power.

\subsubsection*{Example: Qubit unitaries}
Let us now focus on implementing the simplest kind of channels, 
qubit unitary gates.
A unique representation of a qubit unitary up to incoherent symmetries
is given by	\cite{bendana2017}
\begin{equation}
  \label{qub_uni}
  U_\theta=\begin{pmatrix}
             c & -s\\ 
             s & c
           \end{pmatrix},
\end{equation}
where $c = \cos\theta$, $s = \sin \theta$ and $0<\theta\leq \frac{\pi}{4}$.
In this case we can compute all the quantities defined above:
\begin{align}
C_{LR}(U_\theta)&=\log(1+\sin2\theta)=C_{\mathrm{amo}}(U_\theta),\label{qubitAmoco}\\
\simcoz{U_\theta}&=\begin{cases} 0&\theta=0,\\
1&\mathrm{otherwise}.\end{cases},
\end{align}
as can be readily checked by noticing that the robustness of a qubit state is equal to its 
$\ell_1$-norm of coherence (see Eq. (\ref{app:l1nc})): $C_R\big(U_\theta\ket{i}\!\bra{i}U_\theta^\dagger\big)
      =C_{\ell_1}\big(U_\theta\ket{i}\!\bra{i}U_\theta^\dagger\big)
      =2c s = \sin 2\theta$.\\
Eventually, in contrast with the case of IOs \cite{bendana2017}, where the implementation of qubit unitaries consumes the entire cosbit resource, 
Theorem \ref{theoremk-sim} ensures that MIOs do allow for 
coherence recycling. More precisely, there exists a MIO $\cM$ such that 
$\cM(\Psi_2\otimes \rho)=\sigma  \otimes U_\theta\rho U_\theta^\dagger$ 
for any qubit state $\rho$ if and only if the output resource state $\sigma$ has coherence
\begin{equation}
  \label{upBoundCohLeft}
  C_R(\sigma)\leq \dfrac{1-\sin 2\theta}{1+\sin 2\theta}.
\end{equation}

\subsection{Arbitrary resources for MIO implementation of channels}
\label{noncosdits}

Let us now go beyond the assumptions of Sec.~\ref{cosdits} by considering 
a scenario where non-maximally coherent states are 
employed as resources (see~\cite{Stahlke2011} for an analogous
result in entanglement theory).
In particular, we want to study under which conditions the MIO 
implementation of a quantum channel in the settings of 
Fig.~\ref{fig:settingnoancilla} and Fig.~\ref{fig:setting} is still possible.

We begin by introducing a semidefinite program to 
{assess} the performance of an arbitrary resource at 
implementing any target channel; this program
yields the best approximate MIO implementation of the target channel with a given 
resource state (Fig.~\ref{fig:settingnoancilla}), as measured by the 
diamond norm. 
\begin{proposition}
  \label{diamondnorm_sdp}
  The smallest diamond norm error for the implementation of a 
  quantum channel $\cN:A\rightarrow B$ via a MIO $\cM:R\otimes A\rightarrow B$ with 
  a coherent resource $\omega\in R$ is given by the following 
  semidefinite program:
  \begin{equation}\begin{split}
	  \label{diamondsdp2}
	  \min &\ \lambda\\
	   \text{\rm s.t.}&\ J_{\cM}    \text{ is the Choi matrix of } \cM \in \text{MIO}\\
 		&\ J_{\cE} =    \tr_{R}((\omega^t \otimes \mathds{1}_{A} \otimes \mathds{1}_{B}) J_{\cM})\\
 		&\ Z  \geq 0\\
 		&\ Z  \geq J_{\cN}-J_{\cE}\\
 		&\ \lambda \mathds{1}_{A} \geq \tr_{B} Z,
  \end{split}\end{equation}
  where $J_{\cN}$ is the Choi matrix of $\cN$,  $\omega^{t}$ denotes the transpose of $\omega$, and $J_{\cE}$ that of its implementation, 
  $\cE=\tr_{B}\cM(\omega\otimes\cdot)$. Recall that the Choi matrix $J_{\cM}$ of a MIO $\cM$ as in Eq.~\eqref{diamondsdp2}, is fully characterized by $J_{\cM}\geq 0$,
  $ \tr_B J_{\cM} = \mathds{1}_{RA} $ and $ \tr ((\ket{i}\!\bra{i}_{RA}\otimes \ket{j}\!\bra{k}_{B})J_{\cM}) =0\ \forall i \mbox{ and } \forall j\neq k$.
\end{proposition}
In  Appendix~\ref{app:sdps} we provide a simpler semidefinite program for the case of unitary channels, where the precision of the simulation is assessed  in terms of the average gate fidelity $f(U,\cN)$ \cite{horodecki99}, rather than by the diamond distance. Specifically, we compute the entanglement fidelity $F(U,\cN)=\tr J_U J_\cN$ (where $J_U$ and $J_\cN$ are the Choi matrices of $U$ and $\cN$, respectively), which fulfills $F(U,\cN)=\dfrac{(d+1)f(U,\cN)-1}{d}$, where $d$ is the dimension of the Hilbert space on which the channels act \cite{horodecki99}. From now on we will refer to it as ``gate fidelity".

Regarding instead the recycling setting of Fig.~\ref{fig:setting}, thanks to Theorem~\ref{amortizedc} it is straightforward to write a semidefinite program for the $\epsilon$-error amortized cost of a channel with cosdit input resource, as shown in Fig.~\ref{fig:amoCost}. If, in addition, we ask for the maximum robustness of coherence left in the output resource when a non-maximally coherent input resource is employed, we end up with the following optimization problem:
\begin{proposition}
\label{cohLeftNonLin}
The maximum coherence left in the resource $\sigma\in S$ after the 
  implementation of a quantum channel $\cN:A\rightarrow B$ via a
  MIO $\cM:R\otimes A \rightarrow S\otimes B$ and a coherent resource $\omega\in R$
  up to error $\epsilon$ in diamond norm is:
  \begin{equation}\begin{split}
    \label{cohLeftNonLinEQ}
    \max &\ C_{R}(\sigma)\\
     \text{\rm s.t.}&\ \cM(\omega\otimes\cdot)= \sigma\otimes\cL(\cdot)\\
     &\ \frac{1}{2}||\cN-\cL||_{\diamond}\leq \epsilon.
    \end{split}
    \end{equation}
\end{proposition}
This problem captures exactly the setting of Fig.~\ref{fig:setting} with arbitrary input resource. Note, however, that it cannot be formulated as a semidefinite program, since the tensor-product constraint is non-linear in the optimization variables $\sigma$ and $\cL$. One can devise alternative semidefinite programs by relaxing the constraints in \eqq{cohLeftNonLinEQ}, as we discuss in Appendix~\ref{app:altCohLeft}.

The semidefinite program in Proposition~\ref{diamondnorm_sdp} and that for the amortized cost allow for a thorough numerical analysis of our problem, see Figs~\ref{fig:fidelity_2d_diamond}-\ref{fig:amoCost}.
Nevertheless, before analysing the numerical results for qubits, let us 
present some general, analytical results on MIO implementation of channels with non-maximally coherent resources. As a start, according to Theorem \ref{theorem1}, any resource state that is MIO-convertible into a cosdit of coherence rank $d$ can be trivially used to implement any quantum channel of log-robustness smaller than $\log d$. 
The following lemma provides us with a simple necessary 
condition for the existence of pure state transformation via MIO.
\begin{lemma}
  \label{lemma1}
  Assume that a pure state $\psi$ is transformed to a pure state $\phi$ via MIO, then 
  \begin{equation}
    \label{eq:fidcrit}
    \lambda_1(\psi) \leq \lambda_1(\phi),
  \end{equation}
  where $\lambda_1(\rho)= \max_i \bra{i}\rho\ket{i}$ is the largest diagonal entry 
  of an operator $\rho$, which coincides with its fidelity of 
  coherence~\cite{streltsov2015,shao2015} on pure states.
\end{lemma}
Condition \eqref{eq:fidcrit} is also sufficient if the 
target state is maximally coherent, since if $\lambda_1(\psi)\leq \frac{1}{d}$ then
there exists a IO$\subset$MIO that transforms $\psi$ into $\Psi_{d}$ 
due to the majorization criterion~\cite{winter2016}. 

Hence, for any channel $\cN$ with $C_{LR}(\cN)\leq\log d$ there 
exists a continuous family of resource states of dimension $d$ that 
allow for the exact MIO simulation of it via conversion to cosdits.
We now ask whether there exist any channels that can be implemented using pure 
resource states $\ket{\omega}$ that are not convertible into cosdits, 
i.e., $\lambda_1(\omega)> 1/d$. To this end let us define a special 
class of states that will be useful in this context.
\begin{definition}\label{def:flagpole}
	A $d$-dimensional \emph{flagpole state} is a pure state of the form
	\begin{equation}
	\label{eq:flagpole}
		\ket{\varphi_p} = \sqrt{p}\ket{0}+\sqrt{\frac{1-p}{d-1}}\sum_{j=1}^{d-1}\ket{j},\ \text{with} \ \ \frac1d \leq p \leq 1.
	\end{equation}
\end{definition}
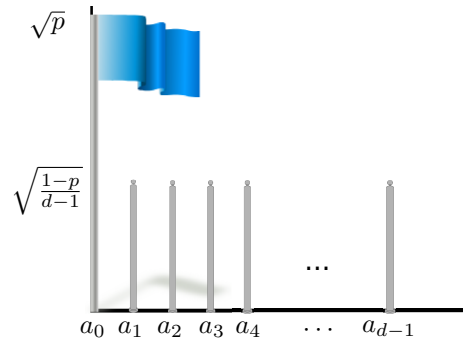
\begin{figure}[t!]
	\centering
	\begin{tikzpicture}
 	\node[inner sep=0pt] at (0,0)
	{\includegraphics[width = 5cm]{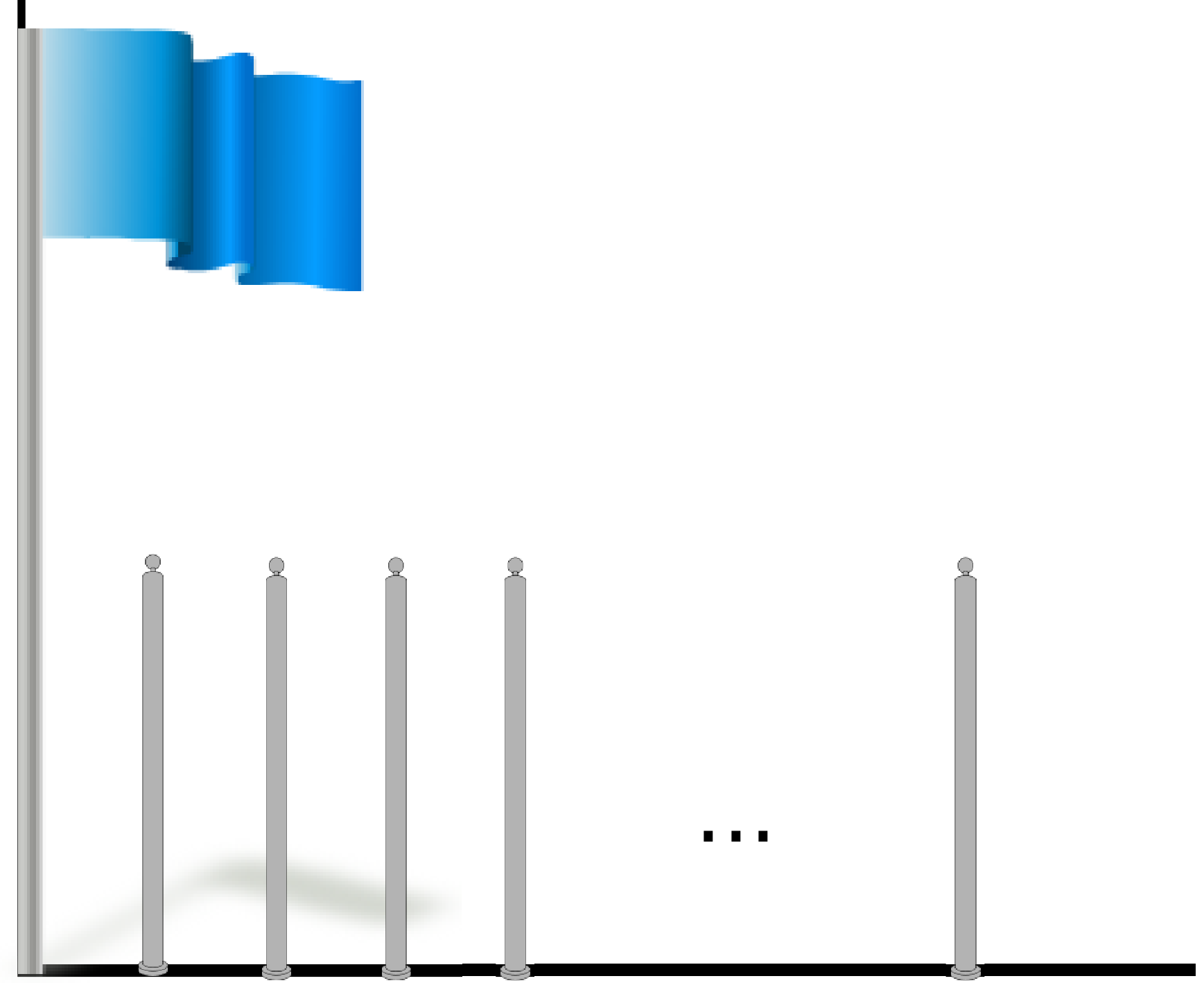}};
	\node[] at (-2.4,-2.25) {\small $a_0$};
	\node[] at (-1.9,-2.25) {\small $a_1$};
	\node[] at (-1.38,-2.25) {\small $a_2$};
	\node[] at (-0.83,-2.25) {\small $a_3$};
	\node[] at (-0.35,-2.25) {\small $a_4$};
	\node[] at (0.6,-2.25) {\small $\cdots$};
	\node[] at (1.5,-2.25) {\small $a_{d-1}$};
	\node[] at (-3,1.8) {\small $\sqrt{p}$};
	\node[] at (-3,-0.4) {\small $\sqrt{\frac{1-p}{d-1}}$};
	% \draw (-4,-3) to[grid with coordinates] (4,3);
\end{tikzpicture}
	\caption{Schematic depiction of the coefficients $\{a_j\}_{j=0}^{d-1}$ of a flagpole state $\ket{\varphi_p}$ in the incoherent basis.}
	\label{fig:flagpolestate}
\end{figure}
The structure of flagpole states, shown in Fig.~\ref{fig:flagpolestate}, endows them with several useful properties.
Thanks to the majorization criterion~\cite{winter2016}, namely that the pure state transformation $\phi \xrightarrow{IO}\psi$ is possible if $\cal D (\psi) \succ  \cal D (\phi)$ (where $\cal D$ is the complete dephasing map), it is easy to show that: i) for all pure states $\phi$ we can transform $\phi\mapsto\varphi_p$ with a specific value of $p\geq F_C(\phi)$ (see Eq. (\ref{fidofcoh})); ii) conversely, we can transform $\varphi_p\mapsto\phi$ for all $\phi$ such that $p\leq F_C(\phi)$.
{In other words, $\varphi_p$ is the most coherent state of fixed coherence rank
with fidelity of coherence larger than or equal to $p$.
Moreover, as explained above, a $d$-dimensional flagpole state with 
$p\leq \frac{1}{d-1}$ can be converted into a $(d-1)$-dimensional 
cosdit and thus trivially implements any channel of log-robustness 
smaller or equal than $\log(d-1)$.  
For all other flagpoles with $p>\frac{1}{d-1}$,
which cannot be converted to a cosdit of rank $d-1$, 
the following theorem holds.
\begin{theorem}\label{theorem2}
	For any quantum channel $\cN:A\rightarrow B$, if
	\begin{equation}\label{p_achievability}
	  p\leq \frac{1}{1+C_R(\cN)},
	\end{equation}
	then there exists a MIO $\cM:R\otimes A\rightarrow B$ such that
	\mbox{$\cM(\varphi_p \otimes \rho) = \cN(\rho)$} for all states $\rho$, 
	where $\varphi_p$ is a $d$-dimensional flagpole state, $d > |B|$.
\end{theorem}

This proves that any pure resource state in dimension larger than 2 is useful for the exact implementation of some coherent unitary channel via MIO. Indeed, any such state can be converted to a flagpole, which in turn can be used for the implementation of a coherent channel, in particular a unitary. Note also that in general the bound on $p$ in Theorem~\ref{theorem2} does not single out all flagpoles that can implement a given channel, since it relies on a specific ansatz (see Appendix~\ref{app:flagpoles} for details). 

\subsubsection*{Example: Qubit unitaries}

Analogously to the previous section, we now address the implementation of 
qubit unitaries when non-maximally coherent resources are available. 

\begin{figure}[t!]
 	\centering
 	\begin{tikzpicture}
 	\node[inner sep=0pt] at (0,0)
	{\includegraphics[width = 0.8\linewidth]{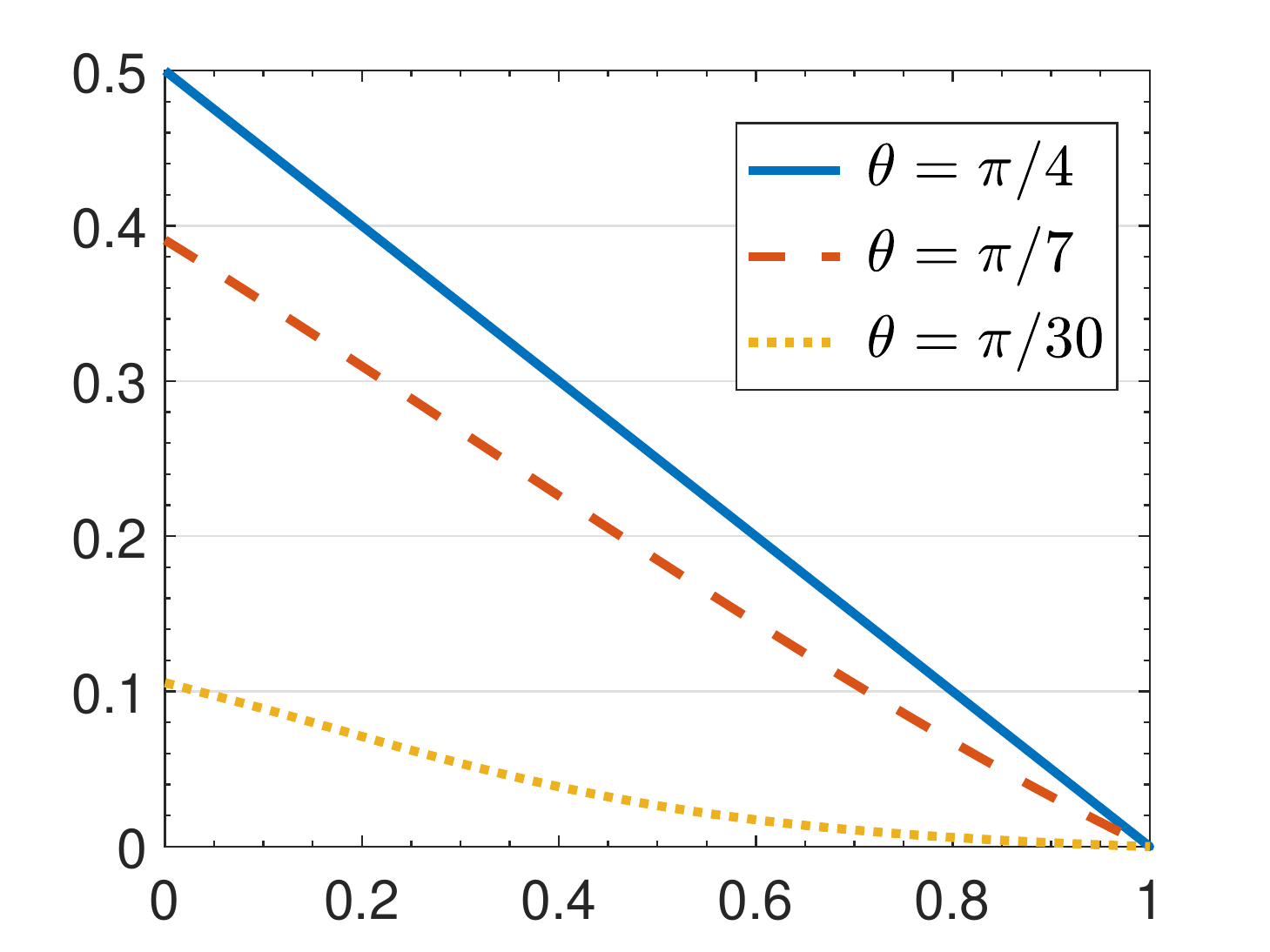}};
	\node[] at (0.2,-2.7) {\small $C_R(\omega)$};
	\node[rotate=90] at (-3.3,0) {\small Diamond norm error};
\end{tikzpicture}
 	\caption{Minimum diamond norm distance between the qubit unitary 
	         $U_\theta$ and a MIO implementation of it 
	         with pure qubit resource state $\ket{\omega}$ vs. robustness of 
	         coherence of the latter. For $\theta=\pi/30$ (dotted orange line), 
	         $\theta=\pi/7$ (dashed red line) and $\theta=\pi/4$ (solid blue line). 
	         Note that an exact implementation is possible only with a 
	         cosbit resource state. 
	         }
 	\label{fig:fidelity_2d_diamond}
\end{figure}

The first question we want to raise is whether it is possible to use a 
non-maximally coherent pure qubit resource in order to implement a qubit 
unitary, even one that generates little coherence. In~\cite{bendana2017} 
it was proven that, if the free operations are IOs, this is only possible 
with a cosbit, no matter how coherent the unitary is. 
This is still the case under MIOs, even for higher-dimensional unitaries, 
when restricting to qubit resource states:

\begin{proposition}\label{proposition_qubit_qubit}
 	The only pure qubit resource state $\ket{\omega}\in \mathbb{C}^2$ that 
	permits the MIO implementation of some unitary gate of arbitrary dimension is
 	the cosbit.
\end{proposition}

\begin{figure}[t!]
\centering
\includegraphics[width = 0.9\linewidth]{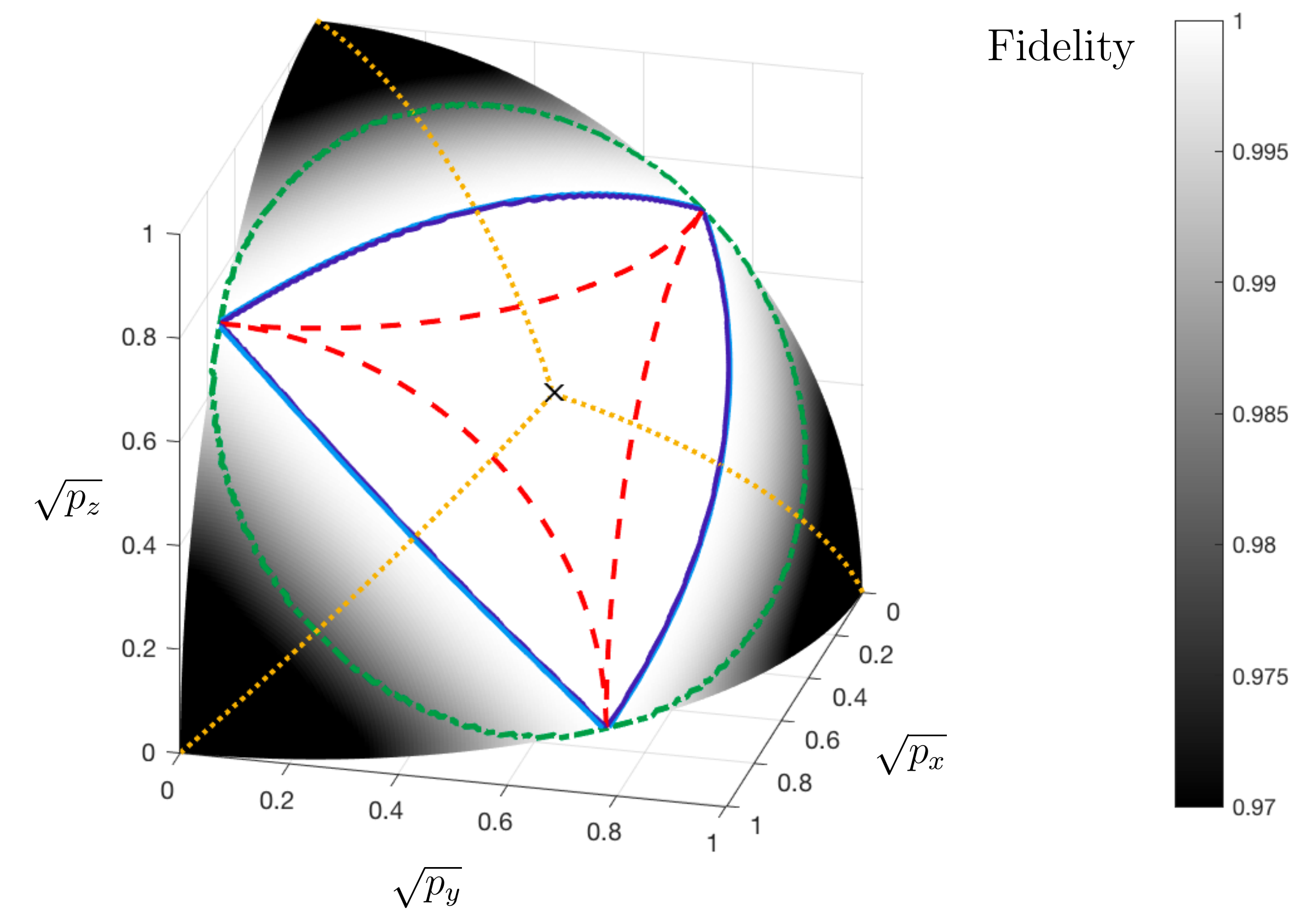}

%\begin{tikzpicture}
 %   \node[inner sep=0pt] at (0,0)
  %  {\includegraphics[width = 0.9\linewidth]{figure4}};
   % \node[rotate = -10] at (1.8,-2.2) {\footnotesize $\sqrt{p_x}$};
    %\node[rotate=-10] at (-1.3,-2.6) {\footnotesize $\sqrt{p_y}$};
    %\node[rotate= 0] at (-3.6,-0.4) {\footnotesize $\sqrt{p_z}$};
    %\node[rotate= 0] at (1.85,2.35) {\small Fidelity};
    % \draw (-4,-3) to[grid with coordinates] (4,3);
%\end{tikzpicture}
\caption{Gate fidelity for the MIO implementation of a qubit unitary $U_\theta$ with $\theta=\pi/14$ using a generic qutrit ($d=3$) resource 
\mbox{$\ket{\psi}=\sqrt{p_{x}}\ket{0}+\sqrt{p_{y}}\ket{1}+\sqrt{p_{z}}\ket{2}$}.
The central cross ($\times$) represents the maximally coherent state. The dashed red line around it encloses states that can be transformed with MIO to cosbits, the blue-solid line encloses the states that allow for an exact implementation ($F=1$), while the 
 dashed-dotted green line encloses all states that cannot be obtained via MIO by diluting a cosbit.
 The dotted yellow lines represent the family of flagpole states.}
\label{fig:qutrit}
\end{figure}

As an illustration of this fact, Fig.~\ref{fig:fidelity_2d_diamond} shows that only a cosbit resource allows for an exact implementation of the qubit unitary $U_\theta$, for several values of $\theta$. Input coherence can be saved only at the expense of allowing for an approximate implementation of the gate. 
Moreover, in Fig.~\ref{fig:qutrit} we give a full characterization of the optimal performance of a MIO simulation, as measured by the gate fidelity, for general qutrit resource states.
Without loss of generality we can focus on the upper region of the plot, $p_{x}, p_{y}\leq p_{z}=\lambda_{1}(\psi)$. Qubit resources are found in the planes defined by $p_{x}=p_y=0$, where perfect simulation ($F=1$) is only reached for cosbits,  $p_{z}=1/2$. The red dashed line delimits the qutrit states that can be distilled into a cosbit, $p_{z}\leq 1/2$, and hence also attain $F=1$.  However, 
perfect simulation can also be attained with other qutrit states: those that fall below the solid blue line in Fig.~\ref{fig:qutrit}. In particular, the qutrit state with the highest value of $p_{z}$, i.e. the least coherent qutrit as measured by the fidelity of coherence, that allows for perfect simulation is a flagpole. Indeed, this agrees with the predictions of Theorem \ref{theorem2}:  
\begin{proposition}\label{proposition_qubit_qutrit}
	A $d$-dimensional flagpole resource state $\ket{\varphi_p}$ allows for the MIO implementation of a qubit unitary $U_\theta$ if
	\begin{equation}\label{pqutrit}
	p\leq \dfrac{1}{1+\sin 2\theta},
	\end{equation}
	where $d\geq 3$ and $0<\theta \leq \frac{\pi}{4}$. Furthermore, for a qutrit flagpole state  $\ket{\varphi_p}\in \mathbb{C}^3$ this is also the highest allowed value of $p$.
\end{proposition}
In Fig.~\ref{fig:qutrit} all flagpole states are identified by a yellow dotted line that interpolates between the incoherent state $\ket{0}$ and the cosbit.
Finally, the set of states that cannot be obtained  via MIO from a cosbit, 
 are enclosed by the dashed-dotted green line which is determined by the intersection of the sphere of qutrit states (in the positive octant) and the plane $\braket{\Psi_{3}}{\psi}=1$, as 
shown in \eqq{planeconv}. 
Extensive numerical evidence shows that the blue solid line that delimites the region of states that enable an exact MIO implementation is also given by the intersection of the qutrit sphere with a plane of the form $\braket{\Phi_{\theta}}{\psi}\leq f(\theta)$, where the constant $f(\theta)$ and normal vector $\ket{\Phi_{\theta}}$ can be analytically found by imposing that the plane includes the cosbit and the flagpole saturating \eqq{pqutrit}.

\begin{figure}[t!]
\centering
\begin{tikzpicture}
    \node[inner sep=0pt] at (0,0)
    {\includegraphics[width = 0.9\linewidth]{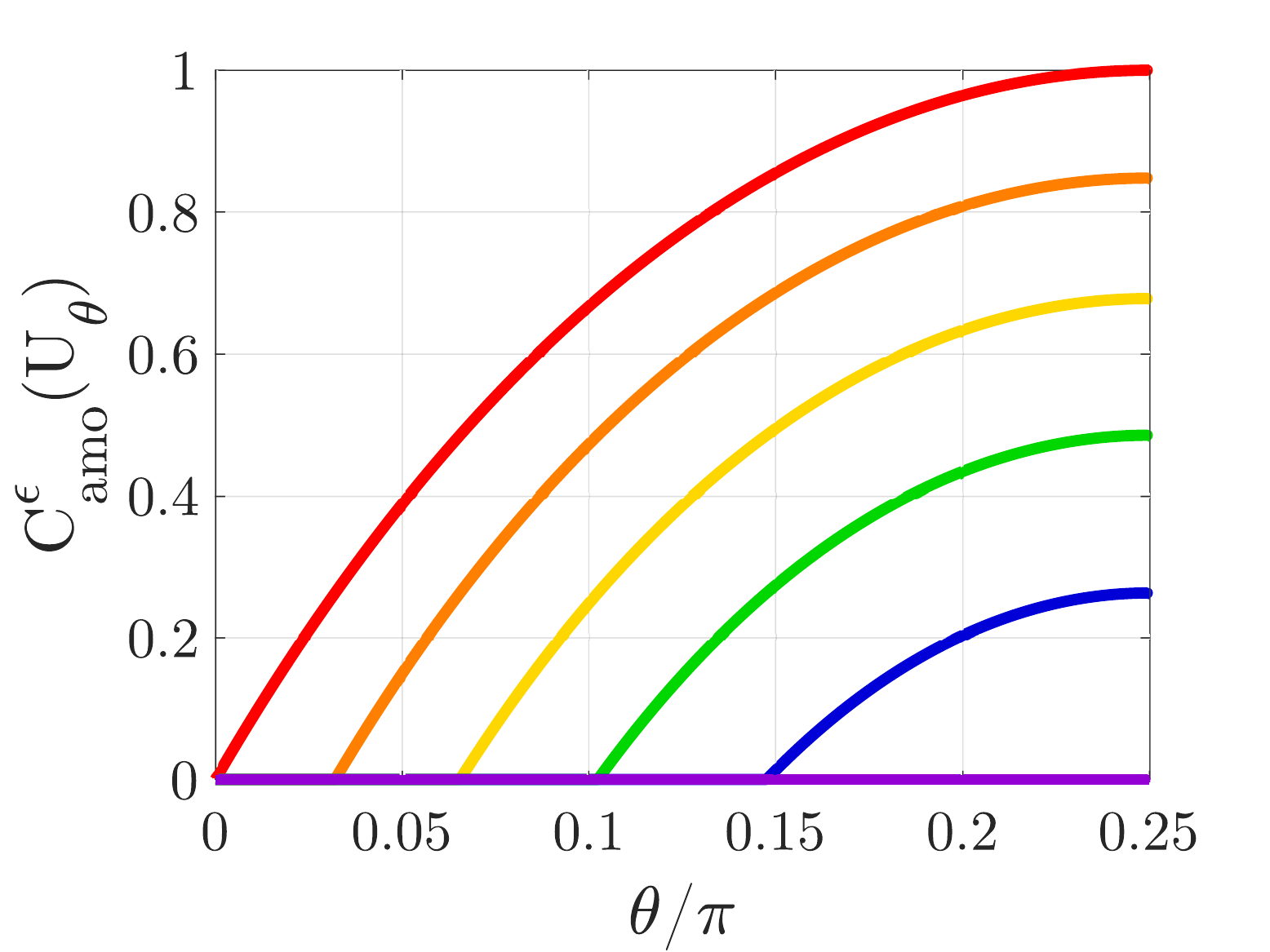}};
   % \node[rotate = -10] at (1.8,-2) {\footnotesize $\sqrt{p_x}$};
    %\node[rotate=-10] at (-1.3,-2.6) {\footnotesize $\sqrt{p_y}$};
    %\node[rotate= 0] at (-3.6,-0.4) {\footnotesize $\sqrt{p_z}$};
    %\node[rotate= 0] at (1.85,2.35) {\small Fidelity};
    % \draw (-4,-3) to[grid with coordinates] (4,3);
\end{tikzpicture}
\caption{Plot of the $\epsilon$-error amortized cost of the qubit unitary $U_\theta$ as vs. $\theta/\pi$, for several error thresholds $\epsilon \in \{0, 0.1, 0.2, 0.3, 0.4, 0.5\}$ (from top to bottom). The amortized cost is higher for more coherent unitaries but it decreases if larger errors are allowed. In particular, when the amortized cost becomes zero it means that the given unitary can be simulated up to the given error with a MIO, which is not necessarily the identity channel, without using any coherence of the input resource state.}
\label{fig:amoCost}
\end{figure}

Regarding the case of coherence recycling after the approximate implementation of a qubit unitary, Fig.~\ref{fig:amoCost} shows the amortized cost of qubit unitaries for several error thresholds. As expected, the amortized cost increases with the coherence of the target unitary, since less coherence can be recycled when implementing more coherent unitaries. Moreover, if we allow a larger error threshold to implement the same unitary, it is possible to obtain more coherence back for the same input resource at the cost of a worse approximation. When the amortized cost becomes zero, Theorem~\ref{amortizedc} implies that there exists a MIO that implements the target unitary without consuming any coherence. This happens if and only if the simulation cost of the unitary for the same error threshold is zero.

\subsection{Comparison with coherence distillation \protect\\ and the reversibility problem}
\label{distillation}
In this work, we have focused on the cost of implementing
a channel, but in \cite{bendana2017} this was already 
considered under IOs in conjunction with the capacity to create pure coherence
from the given resource channel. This is an important point, because
it touches upon the issue of (asymptotic) reversibility of the 
resource theory. Namely, since in the theory of coherence there exists a unit resource, the cosbit, it is possible to compare the minimum 
rate of cosbits needed to implement the channel, and the
maximum rate of cosbits extractable from the given channel: 
the theory is asymptotically reversible if and only if these
two numbers coincide for all channels.
For the latter problem, we find in Ref.~\cite{bendana2017} a 
developed theory, even if it was worked out with IOs as free
operations. However, the changes required to the definitions,
protocols and bounds in \cite[Sec.~II]{bendana2017} (especially
to its Thm.~1) are minimal and readily result in the following
theorem. Before stating it, let us first recall the generic protocol for coherence generation with many uses of a channel $\cN:A\rightarrow B$: i) an initial incoherent state $\rho_0\in AC$ undergoes the action of the channel, then is post-processed by a MIO $M_1:BC\rightarrow AC_1$, obtaining the state $\rho_1=M_1\circ(N\otimes\id)(\rho_0)$; ii) the same process is iterated $n-1$ times, producing the intermediate states $\rho_j=M_j\circ(N\otimes\id)(\rho_{j-1})$, $j=1,\cdots,n$; iii) the iterations stop when the final state $\rho_n$ is $\epsilon$-close to the cosdit $\Psi_{2}^{\otimes nR}$, i.e., $\tr(\rho_n\Psi_{2}^{\otimes nR})\geq1-\epsilon$,
with $R$ being the resulting coherence generating rate. Then the coherence-generating capacity of $\cN$ is defined as follows:
\begin{equation}
C_{gen, MIO}^{(\infty)}(\cN):=\sup_{\epsilon>0}\limsup_{n\rightarrow\infty}R.
\end{equation}
\begin{theorem}\label{asyCohGen}
  The asymptotic coherence-generating capacity of a channel
  $\cN:A\rightarrow B$ is given by
  \begin{align}
    C_{\text{gen},\text{MIO}}^{(\infty)}(\cN) &= \sup_{\rho_{AC}} C_r((\cN\otimes\id_C)(\rho_{AC}))\nonumber \\
    &-C_r(\rho_{AC})\\
                         & = \sup_C \widehat{P}_r(\cN\otimes \id_C)
                          = \mathbb{P}_r(\cN),
\end{align}
  where the suprema are over all auxiliary systems $C$ and
  the first additionally over states $\rho_{AC}$ on $A\otimes C$.
  The second formula features the \emph{relative entropy cohering power} \cite{garcia2016,bendana2017}
  $\widehat{P}_r(\cN):=\max_\rho C_r(\cN(\rho))-C_r(\rho)$ ; $\mathbb{P}_r(\cN):=\sup_k \widehat{P}_r(\cN\otimes \id_k)$ was
  introduced in \cite{bendana2017} as \emph{complete relative
  entropy cohering power}.
  \hfill$\blacksquare$
\end{theorem}
The coherence-generating capacity, despite its nice information theoretic
formula, is by no means an easy quantity to compute. Notwithstanding
the explicit expression for the relative entropy of coherence, the maximization over
the state $\rho_{AC}$ is not necessarily well-behaved. Also, we do not know
how large an auxiliary system $C$ is required, or indeed if any at all.

On the other hand, the asymptotic cost of implementing the channel
can be expressed using our above results, in particular
Theorem~\ref{theorem1}. Namely, it comes down to 
\begin{equation}
  \label{channel-AEP}
  C_{\text{sim},\text{MIO}}^{(\infty)}(\cN) 
     = \sup_{\epsilon > 0} \limsup_{n\rightarrow\infty}
               \frac1n \log (1+C_R^\epsilon(\cN^{\otimes n})).
\end{equation}

There is one special case in which we know what this limit is,
namely when for all $\rho$, $\cN(\rho) = \sigma$, i.e.~$\cN$ is a 
constant channel. Then, $C_R(\cN) = C_R(\sigma)$, and the right
hand side of Eq.~(\ref{channel-AEP}) converges~\cite{brandao2015a} to the relative
entropy of coherence, $C_r(\sigma)$. More generally, for a
cq-channel $\cN(\ket{i}\!\bra{j}) = \delta_{ij}\sigma_i$, written
in the incoherent basis of the input state, the same reasoning
yields
\begin{align}
  C_{\text{sim},\text{MIO}}^{(\infty)}(\cN) &= C_{\text{gen},\text{MIO}}^{(\infty)}(\cN) \nonumber \\
  &= \max_i C_r(\cN(\ket{i}\!\bra{i})) \nonumber,
  \end{align}
i.e., the theory is asymptotically reversible when restricted to cq-channels. It is not known under which conditions this happens for general channel resource theories, since the results of~\cite{brandao2015a} cannot be directly applied to them.

\section{Discussion}
In this work we studied the conversion of static coherence into dynamic coherence under the most general framework, which employs MIOs as the free operations. 
To that end, we have introduced the robustness of coherence of channels, which includes as a special case the corresponding measure for states, showing its operational meaning in two respects: i) the $\log$-robustness equals the cohering power of the channel; ii) when maximally coherent states are employed as resources, the robustness determines the implementation cost of the channel and the $\log$-robustness characterizes the amortized implementation cost with recycling of coherence at the output, as well as the asymptotic cost of realizing exactly many independent instances of the channel. 

We have then considered the more general case of implementation with arbitrary resource states and provided a semidefinite program to find the best approximation to any given channel with any resource. In the same setting, we have analytically demonstrated that any state in dimension larger than 2, however weakly so, is useful for the exact implementation of some coherent unitary channel, owing to the fact that every such state can be converted into a flagpole state and that, in turn, every flagpole is able to realize some channel. 

Throughout the article, we have considered the specific case of qubit unitary implementation and studied it in detail, showing a direct application of our general results. Finally, we have related our findings to the problem of asymptotic reversibility, by introducing the asymptotic coherence generating capacity and the asymptotic cost of implementation of a given channel under MIOs. Since both quantities are not easily computable, it remains to be known whether they coincide, i.e. whether the theory is asymptotically reversible. 

Our findings significantly advance the state of the art on the theory of coherence of channels, providing an operational measure for its quantification and a good starting point for future research on the subject. In particular, some open questions immediately follow from our results: i) whether a closed expression for the asymptotic implementation cost of channels can be obtained from the smoothed robustness of coherence; ii) whether the coherence left after the approximate implementation of a channel with an arbitrary input resource can be formulated as an efficient optimization problem; iii) whether the set of pure resource states that allow for the exact MIO implementation of an arbitrary unitary is always determined by a linear function of their coefficients in the incoherent basis.

Finally, our results may also spur investigation on other lines of research such as, e.g., exploring  similar approaches in the resource theories of non-classicality and athermality, analysing the connection of our findings with those in entanglement theory \cite{berta2013,eisert2000}, studying the applicability of coherence theory in an experimental setting and extending further the general framework of channel resource theories with the help of quantum combs \cite{chiribella2008}.

\section*{Acknowledgements}
The authors thank Bartosz Regula, Zi-Wen Liu, Nilanjana Datta and Andrea R. Blanco for interesting discussions on various aspects of the present work.
The authors acknowledge support from Spanish MINECO, project FIS2016-80681-P with the support of AEI/FEDER funds; the Generalitat de Catalunya, project CIRIT 2017-SGR-1127. MGD is supported by a doctoral studies fellowship of the Fundaci\'{o}n ``la Caixa''. MS is supported by the Spanish MINECO, project IJCI-2015-24643.

\bibliographystyle{plainnat}
\bibliography{MIO}

\begin{thebibliography}{69}
\providecommand{\natexlab}[1]{#1}
\providecommand{\url}[1]{\texttt{#1}}
\expandafter\ifx\csname urlstyle\endcsname\relax
  \providecommand{\doi}[1]{doi: #1}\else
  \providecommand{\doi}{doi: \begingroup \urlstyle{rm}\Url}\fi

\bibitem[{\AA}berg(2006)]{aberg2006}
Johan {\AA}berg.
\newblock Quantifying superposition.
\newblock arXiv:quant-ph/0612146, 2006.

\bibitem[Aharonov et~al.(1998)Aharonov, Kitaev, and Nisan]{aharonov1998}
Dorit Aharonov, Alexei Kitaev, and Noam Nisan.
\newblock Quantum circuits with mixed states.
\newblock In \emph{Proceedings of the Thirtieth Annual ACM Symposium on Theory
  of Computing}, STOC '98, pages 20--30, New York, NY, USA, 1998. ACM.
\newblock ISBN 0-89791-962-9.
\newblock \doi{10.1145/276698.276708}.
\newblock URL \url{http://doi.acm.org/10.1145/276698.276708}.

\bibitem[Anand and Pati(2016)]{anand2016}
Namit Anand and Arun~Kumar Pati.
\newblock Coherence and entanglement monogamy in the discrete analogue of
  analog {G}rover search.
\newblock arXiv[quant-ph]:1611.04542, 2016.

\bibitem[Baumgratz et~al.(2014)Baumgratz, Cramer, and Plenio]{baumgratz2014}
T.~Baumgratz, M.~Cramer, and M.~B. Plenio.
\newblock Quantifying coherence.
\newblock \emph{Phys. Rev. Lett.}, 113:\penalty0 140401, Sep 2014.
\newblock \doi{10.1103/PhysRevLett.113.140401}.
\newblock URL \url{https://link.aps.org/doi/10.1103/PhysRevLett.113.140401}.

\bibitem[Ben~Dana et~al.(2017)Ben~Dana, Garc\'{\i}a~D\'{\i}az, Mejatty, and
  Winter]{bendana2017}
Khaled Ben~Dana, Mar\'{\i}a Garc\'{\i}a~D\'{\i}az, Mohamed Mejatty, and Andreas
  Winter.
\newblock Resource theory of coherence: Beyond states.
\newblock \emph{Phys. Rev. A}, 95:\penalty0 062327, Jun 2017.
\newblock \doi{10.1103/PhysRevA.95.062327}.
\newblock URL \url{https://link.aps.org/doi/10.1103/PhysRevA.95.062327}.

\bibitem[Berta et~al.(2013)Berta, Brandão, Christandl, and Wehner]{berta2013}
M.~Berta, F.~G. S.~L. Brandão, M.~Christandl, and S.~Wehner.
\newblock Entanglement cost of quantum channels.
\newblock \emph{IEEE Transactions on Information Theory}, 59\penalty0
  (10):\penalty0 6779--6795, Oct 2013.
\newblock ISSN 0018-9448.
\newblock \doi{10.1109/TIT.2013.2268533}.

\bibitem[Biswas et~al.(2017)Biswas, Garc{\'\i}a~D{\'\i}az, and
  Winter]{biswas2017}
Tanmoy Biswas, Mar{\'\i}a Garc{\'\i}a~D{\'\i}az, and Andreas Winter.
\newblock Interferometric visibility and coherence.
\newblock \emph{Proceedings of the Royal Society of London A: Mathematical,
  Physical and Engineering Sciences}, 473\penalty0 (2203), 2017.
\newblock ISSN 1364-5021.
\newblock \doi{10.1098/rspa.2017.0170}.
\newblock URL
  \url{http://rspa.royalsocietypublishing.org/content/473/2203/20170170}.

\bibitem[Brand\~ao and Gour(2015)]{brandao2015a}
Fernando G. S.~L. Brand\~ao and Gilad Gour.
\newblock Reversible framework for quantum resource theories.
\newblock \emph{Phys. Rev. Lett.}, 115:\penalty0 070503, Aug 2015.
\newblock \doi{10.1103/PhysRevLett.115.070503}.
\newblock URL \url{https://link.aps.org/doi/10.1103/PhysRevLett.115.070503}.

\bibitem[Brand\~ao et~al.(2013)Brand\~ao, Horodecki, Oppenheim, Renes, and
  Spekkens]{brandao2013}
Fernando G. S.~L. Brand\~ao, Micha\l{} Horodecki, Jonathan Oppenheim, Joseph~M.
  Renes, and Robert~W. Spekkens.
\newblock Resource theory of quantum states out of thermal equilibrium.
\newblock \emph{Phys. Rev. Lett.}, 111:\penalty0 250404, Dec 2013.
\newblock \doi{10.1103/PhysRevLett.111.250404}.
\newblock URL \url{https://link.aps.org/doi/10.1103/PhysRevLett.111.250404}.

\bibitem[Brand{\~a}o et~al.(2015)Brand{\~a}o, Horodecki, Ng, Oppenheim, and
  Wehner]{brandao2015b}
Fernando Brand{\~a}o, Micha{\l} Horodecki, Nelly Ng, Jonathan Oppenheim, and
  Stephanie Wehner.
\newblock The second laws of quantum thermodynamics.
\newblock \emph{Proceedings of the National Academy of Sciences of the United
  States of America}, 112\penalty0 (11):\penalty0 3275--3279, 03 2015.
\newblock \doi{10.1073/pnas.1411728112}.
\newblock URL \url{http://www.ncbi.nlm.nih.gov/pmc/articles/PMC4372001/}.

\bibitem[Braun and Georgeot(2006)]{BraunGeorgeot}
Daniel Braun and Bertrand Georgeot.
\newblock Quantitative measure of interference.
\newblock \emph{Phys. Rev. A}, 73:\penalty0 022314, Feb 2006.
\newblock \doi{10.1103/PhysRevA.73.022314}.
\newblock URL \url{https://link.aps.org/doi/10.1103/PhysRevA.73.022314}.

\bibitem[Bu and Xiong(2017)]{bu2017b}
Kaifeng Bu and Chunhe Xiong.
\newblock A note on cohering power and de-cohering power.
\newblock \emph{Quantum Info. Comput.}, 17\penalty0 (13-14):\penalty0
  1206--1220, November 2017.
\newblock ISSN 1533-7146.
\newblock \doi{10.26421/QIC17.13-14}.

\bibitem[Bu et~al.(2017)Bu, Kumar, Zhang, and Wu]{bu2017a}
Kaifeng Bu, Asutosh Kumar, Lin Zhang, and Junde Wu.
\newblock Cohering power of quantum operations.
\newblock \emph{Phys. Lett. A}, 381\penalty0 (19):\penalty0 1670 -- 1676, 2017.
\newblock ISSN 0375-9601.
\newblock \doi{10.1016/j.physleta.2017.03.022}.
\newblock URL
  \url{http://www.sciencedirect.com/science/article/pii/S0375960117302621}.

\bibitem[Buluta and Nori(2009)]{buluta2009}
Iulia Buluta and Franco Nori.
\newblock Quantum simulators.
\newblock \emph{Science}, 326\penalty0 (5949):\penalty0 108--111, 2009.
\newblock ISSN 0036-8075.
\newblock \doi{10.1126/science.1177838}.
\newblock URL \url{http://science.sciencemag.org/content/326/5949/108}.

\bibitem[Chin et~al.(2013)Chin, Prior, Rosenbach, Caycedo-Soler, Huelga, and
  Plenio]{chin2013}
A.~W. Chin, J.~Prior, R.~Rosenbach, F.~Caycedo-Soler, S.~F. Huelga, and M.~B.
  Plenio.
\newblock The role of non-equilibrium vibrational structures in electronic
  coherence and recoherence in pigment--protein complexes.
\newblock \emph{Nature Physics}, 9:\penalty0 113 EP --, 01 2013.
\newblock \doi{10.1038/nphys2515}.

\bibitem[Chiribella et~al.(2008)Chiribella, D'Ariano, and
  Perinotti]{chiribella2008}
G.~Chiribella, G.~M. D'Ariano, and P.~Perinotti.
\newblock Quantum circuit architecture.
\newblock \emph{Phys. Rev. Lett.}, 101:\penalty0 060401, Aug 2008.
\newblock \doi{10.1103/PhysRevLett.101.060401}.
\newblock URL \url{https://link.aps.org/doi/10.1103/PhysRevLett.101.060401}.

\bibitem[Chitambar and Gour(2016)]{chitambar2016a}
Eric Chitambar and Gilad Gour.
\newblock Critical examination of incoherent operations and a physically
  consistent resource theory of quantum coherence.
\newblock \emph{Phys. Rev. Lett.}, 117:\penalty0 030401, Jul 2016.
\newblock \doi{10.1103/PhysRevLett.117.030401}.
\newblock URL \url{https://link.aps.org/doi/10.1103/PhysRevLett.117.030401}.

\bibitem[Chitambar and Hsieh(2016)]{chitambar2016b}
Eric Chitambar and Min-Hsiu Hsieh.
\newblock Relating the resource theories of entanglement and quantum coherence.
\newblock \emph{Phys. Rev. Lett.}, 117:\penalty0 020402, Jul 2016.
\newblock \doi{10.1103/PhysRevLett.117.020402}.
\newblock URL \url{https://link.aps.org/doi/10.1103/PhysRevLett.117.020402}.

\bibitem[Degen et~al.(2017)Degen, Reinhard, and Cappellaro]{degen2017}
C.~L. Degen, F.~Reinhard, and P.~Cappellaro.
\newblock Quantum sensing.
\newblock \emph{Rev. Mod. Phys.}, 89:\penalty0 035002, Jul 2017.
\newblock \doi{10.1103/RevModPhys.89.035002}.
\newblock URL \url{https://link.aps.org/doi/10.1103/RevModPhys.89.035002}.

\bibitem[Einstein et~al.(1935)Einstein, Podolsky, and Rosen]{einstein1934}
A.~Einstein, B.~Podolsky, and N.~Rosen.
\newblock Can quantum-mechanical description of physical reality be considered
  complete?
\newblock \emph{Phys. Rev.}, 47:\penalty0 777--780, May 1935.
\newblock \doi{10.1103/PhysRev.47.777}.
\newblock URL \url{https://link.aps.org/doi/10.1103/PhysRev.47.777}.

\bibitem[Eisert et~al.(2000)Eisert, Jacobs, Papadopoulos, and
  Plenio]{eisert2000}
J.~Eisert, K.~Jacobs, P.~Papadopoulos, and M.~B. Plenio.
\newblock Optimal local implementation of nonlocal quantum gates.
\newblock \emph{Phys. Rev. A}, 62:\penalty0 052317, Oct 2000.
\newblock \doi{10.1103/PhysRevA.62.052317}.
\newblock URL \url{https://link.aps.org/doi/10.1103/PhysRevA.62.052317}.

\bibitem[Engel et~al.(2007)Engel, Calhoun, Read, Ahn, Man{\v c}al, Cheng,
  Blankenship, and Fleming]{engel2007}
Gregory~S. Engel, Tessa~R. Calhoun, Elizabeth~L. Read, Tae-Kyu Ahn, Tom{\'a}{\v
  s} Man{\v c}al, Yuan-Chung Cheng, Robert~E. Blankenship, and Graham~R.
  Fleming.
\newblock Evidence for wavelike energy transfer through quantum coherence in
  photosynthetic systems.
\newblock \emph{Nature}, 446:\penalty0 782 EP --, 04 2007.
\newblock \doi{10.1038/nature05678}.

\bibitem[Faist et~al.(2015{\natexlab{a}})Faist, Dupuis, Oppenheim, and
  Renner]{faist2015a}
Philippe Faist, Fr{\'e}d{\'e}ric Dupuis, Jonathan Oppenheim, and Renato Renner.
\newblock The minimal work cost of information processing.
\newblock \emph{Nature Communications}, 6:\penalty0 7669 EP --, 07
  2015{\natexlab{a}}.
\newblock \doi{10.1038/ncomms8669}.

\bibitem[Faist et~al.(2015{\natexlab{b}})Faist, Oppenheim, and
  Renner]{faist2015b}
Philippe Faist, Jonathan Oppenheim, and Renato Renner.
\newblock Gibbs-preserving maps outperform thermal operations in the quantum
  regime.
\newblock \emph{New J. Phys.}, 17\penalty0 (4):\penalty0 043003,
  2015{\natexlab{b}}.
\newblock \doi{10.1088/1367-2630/17/4/043003}.

\bibitem[Feynman et~al.(1965)Feynman, Leighton, and Sands]{feynman1965}
R.~F. Feynman, R.~B. Leighton, and M.~Sands.
\newblock \emph{Feynman Physics Lectures}, volume~3.
\newblock Addison-Wesley Publishing Company, 1965.
\newblock ISBN 0-201-02118-8-P.

\bibitem[Galindo and Mart\'{\i}n-Delgado(2002)]{galindo2002}
A.~Galindo and M.~A. Mart\'{\i}n-Delgado.
\newblock Information and computation: Classical and quantum aspects.
\newblock \emph{Rev. Mod. Phys.}, 74:\penalty0 347--423, May 2002.
\newblock \doi{10.1103/RevModPhys.74.347}.
\newblock URL \url{https://link.aps.org/doi/10.1103/RevModPhys.74.347}.

\bibitem[Garc\'{\i}a-D\'{\i}az et~al.(2016)Garc\'{\i}a-D\'{\i}az, Egloff, and
  Plenio]{garcia2016}
Mar\'{\i}a Garc\'{\i}a-D\'{\i}az, Dario Egloff, and Martin~B. Plenio.
\newblock A note on coherence power of n-dimensional unitary operators.
\newblock \emph{Quantum Info. Comput.}, 16\penalty0 (15-16):\penalty0
  1282--1294, November 2016.
\newblock ISSN 1533-7146.
\newblock \doi{10.26421/QIC16.15-16}.

\bibitem[Georgescu et~al.(2014)Georgescu, Ashhab, and Nori]{georgescu2014}
I.~M. Georgescu, S.~Ashhab, and Franco Nori.
\newblock Quantum simulation.
\newblock \emph{Rev. Mod. Phys.}, 86:\penalty0 153--185, Mar 2014.
\newblock \doi{10.1103/RevModPhys.86.153}.
\newblock URL \url{https://link.aps.org/doi/10.1103/RevModPhys.86.153}.

\bibitem[Giorda and Allegra(2018)]{giorda2017}
Paolo Giorda and Michele Allegra.
\newblock Coherence in quantum estimation.
\newblock \emph{Journal of Physics A: Mathematical and Theoretical},
  51\penalty0 (2):\penalty0 025302, 2018.
\newblock \doi{10.1088/1751-8121/aa9808}.

\bibitem[Giovannetti et~al.(2011)Giovannetti, Lloyd, and
  Maccone]{giovanetti2011}
Vittorio Giovannetti, Seth Lloyd, and Lorenzo Maccone.
\newblock Advances in quantum metrology.
\newblock \emph{Nature Photonics}, 5:\penalty0 222 EP --, 03 2011.
\newblock \doi{10.1038/nphoton.2011.35}.

\bibitem[Gisin and Thew(2007)]{gisin2007}
Nicolas Gisin and Rob Thew.
\newblock Quantum communication.
\newblock \emph{Nature Photonics}, 1:\penalty0 165 EP --, 03 2007.
\newblock \doi{10.1038/nphoton.2007.22}.

\bibitem[Gisin et~al.(2002)Gisin, Ribordy, Tittel, and Zbinden]{gisin2002}
Nicolas Gisin, Gr\'egoire Ribordy, Wolfgang Tittel, and Hugo Zbinden.
\newblock Quantum cryptography.
\newblock \emph{Rev. Mod. Phys.}, 74:\penalty0 145--195, Mar 2002.
\newblock \doi{10.1103/RevModPhys.74.145}.
\newblock URL \url{https://link.aps.org/doi/10.1103/RevModPhys.74.145}.

\bibitem[Gour and Spekkens(2008)]{gour2008}
Gilad Gour and Robert~W Spekkens.
\newblock The resource theory of quantum reference frames: manipulations and
  monotones.
\newblock \emph{New J. Phys.}, 10\penalty0 (3):\penalty0 033023, 2008.
\newblock \doi{10.1088/1367-2630/10/3/033023}.

\bibitem[Gour et~al.(2015)Gour, Müller, Narasimhachar, Spekkens, and
  Halpern]{gour2015}
Gilad Gour, Markus~P. Müller, Varun Narasimhachar, Robert~W. Spekkens, and
  Nicole~Yunger Halpern.
\newblock The resource theory of informational nonequilibrium in
  thermodynamics.
\newblock \emph{Physics Reports}, 583:\penalty0 1 -- 58, 2015.
\newblock ISSN 0370-1573.
\newblock \doi{https://doi.org/10.1016/j.physrep.2015.04.003}.

\bibitem[Hillery(2016)]{hillery2016}
Mark Hillery.
\newblock Coherence as a resource in decision problems: The deutsch-jozsa
  algorithm and a variation.
\newblock \emph{Phys. Rev. A}, 93:\penalty0 012111, Jan 2016.
\newblock \doi{10.1103/PhysRevA.93.012111}.
\newblock URL \url{https://link.aps.org/doi/10.1103/PhysRevA.93.012111}.

\bibitem[Holevo(2012)]{holevo2012}
Alexander~S. Holevo.
\newblock \emph{{Quantum Systems, Channels, Information}}.
\newblock De Gruyter, Berlin, Boston, jan 2012.
\newblock ISBN 9783110273403.
\newblock \doi{10.1515/9783110273403}.
\newblock URL
  \url{http://www.degruyter.com/view/books/9783110273403/9783110273403/9783110273403.xml}.

\bibitem[Horodecki and Oppenheim(2013)]{horodecki2013}
Micha{\l} Horodecki and Jonathan Oppenheim.
\newblock Fundamental limitations for quantum and nanoscale thermodynamics.
\newblock \emph{Nature Communications}, 4:\penalty0 2059 EP --, 06 2013.
\newblock \doi{10.1038/ncomms3059}.

\bibitem[Horodecki et~al.(1999)Horodecki, Horodecki, and
  Horodecki]{horodecki99}
Micha\l{} Horodecki, Pawe\l{} Horodecki, and Ryszard Horodecki.
\newblock General teleportation channel, singlet fraction, and
  quasidistillation.
\newblock \emph{Phys. Rev. A}, 60:\penalty0 1888--1898, Sep 1999.
\newblock \doi{10.1103/PhysRevA.60.1888}.
\newblock URL \url{https://link.aps.org/doi/10.1103/PhysRevA.60.1888}.

\bibitem[Horodecki et~al.(2009)Horodecki, Horodecki, Horodecki, and
  Horodecki]{horodecki2009}
Ryszard Horodecki, Pawe\l{} Horodecki, Micha\l{} Horodecki, and Karol
  Horodecki.
\newblock Quantum entanglement.
\newblock \emph{Rev. Mod. Phys.}, 81:\penalty0 865--942, Jun 2009.
\newblock \doi{10.1103/RevModPhys.81.865}.
\newblock URL \url{https://link.aps.org/doi/10.1103/RevModPhys.81.865}.

\bibitem[Hu et~al.(2017)Hu, Hu, Wang, Peng, Zhang, and Fan]{hu2017}
M.~L. Hu, X.~Hu, J.~C. Wang, Y.~Peng, Y.~R. Zhang, and H.~Fan.
\newblock Quantum coherence and quantum correlations.
\newblock arXiv[quant-ph]:1703.01852, 2017.

\bibitem[Huelga and Plenio(2013)]{huelga2013}
S.F. Huelga and M.B. Plenio.
\newblock Vibrations, quanta and biology.
\newblock \emph{Contemporary Physics}, 54\penalty0 (4):\penalty0 181--207,
  2013.
\newblock \doi{10.1080/00405000.2013.829687}.
\newblock URL \url{https://doi.org/10.1080/00405000.2013.829687}.

\bibitem[Ishizaki and Fleming(2009)]{ishizaki2009}
Akihito Ishizaki and Graham~R. Fleming.
\newblock Theoretical examination of quantum coherence in a photosynthetic
  system at physiological temperature.
\newblock \emph{Proceedings of the National Academy of Sciences}, 106\penalty0
  (41):\penalty0 17255--17260, 2009.
\newblock ISSN 0027-8424.
\newblock \doi{10.1073/pnas.0908989106}.
\newblock URL \url{http://www.pnas.org/content/106/41/17255}.

\bibitem[Killoran et~al.(2016)Killoran, Steinhoff, and Plenio]{killoran16}
N.~Killoran, F.~E.~S. Steinhoff, and M.~B. Plenio.
\newblock Converting nonclassicality into entanglement.
\newblock \emph{Phys. Rev. Lett.}, 116:\penalty0 080402, Feb 2016.
\newblock \doi{10.1103/PhysRevLett.116.080402}.
\newblock URL \url{https://link.aps.org/doi/10.1103/PhysRevLett.116.080402}.

\bibitem[Korzekwa et~al.(2016)Korzekwa, Lostaglio, Oppenheim, and
  Jennings]{korzekwa2016}
Kamil Korzekwa, Matteo Lostaglio, Jonathan Oppenheim, and David Jennings.
\newblock The extraction of work from quantum coherence.
\newblock \emph{New Journal of Physics}, 18\penalty0 (2):\penalty0 023045,
  2016.
\newblock \doi{10.1088/1367-2630/18/2/023045}.

\bibitem[Ladd et~al.(2010)Ladd, Jelezko, Laflamme, Nakamura, Monroe, and
  O'Brien]{ladd2010}
T.~D. Ladd, F.~Jelezko, R.~Laflamme, Y.~Nakamura, C.~Monroe, and J.~L. O'Brien.
\newblock Quantum computers.
\newblock \emph{Nature}, 464:\penalty0 45 EP --, 03 2010.
\newblock \doi{10.1038/nature08812}.

\bibitem[Lostaglio et~al.(2015)Lostaglio, Jennings, and Rudolph]{lostaglio2015}
Matteo Lostaglio, David Jennings, and Terry Rudolph.
\newblock Description of quantum coherence in thermodynamic processes requires
  constraints beyond free energy.
\newblock \emph{Nature Communications}, 6:\penalty0 6383 EP --, 03 2015.
\newblock \doi{10.1038/ncomms7383}.

\bibitem[Lostaglio et~al.(2017)Lostaglio, Jennings, and Rudolph]{lostaglio2017}
Matteo Lostaglio, David Jennings, and Terry Rudolph.
\newblock Thermodynamic resource theories, non-commutativity and maximum
  entropy principles.
\newblock \emph{New Journal of Physics}, 19\penalty0 (4):\penalty0 043008,
  2017.
\newblock \doi{1367-2630/19/i=4/a=043008}.

\bibitem[Ludlow et~al.(2015)Ludlow, Boyd, Ye, Peik, and Schmidt]{ludlow2015}
Andrew~D. Ludlow, Martin~M. Boyd, Jun Ye, E.~Peik, and P.~O. Schmidt.
\newblock Optical atomic clocks.
\newblock \emph{Rev. Mod. Phys.}, 87:\penalty0 637--701, Jun 2015.
\newblock \doi{10.1103/RevModPhys.87.637}.
\newblock URL \url{https://link.aps.org/doi/10.1103/RevModPhys.87.637}.

\bibitem[Malvezzi et~al.(2016)Malvezzi, Karpat, \ifmmode~\mbox{\c{C}}\else
  \c{C}\fi{}akmak, Fanchini, Debarba, and Vianna]{malvezzi2016}
A.~L. Malvezzi, G.~Karpat, B.~\ifmmode~\mbox{\c{C}}\else \c{C}\fi{}akmak, F.~F.
  Fanchini, T.~Debarba, and R.~O. Vianna.
\newblock Quantum correlations and coherence in spin-1 heisenberg chains.
\newblock \emph{Phys. Rev. B}, 93:\penalty0 184428, May 2016.
\newblock \doi{10.1103/PhysRevB.93.184428}.
\newblock URL \url{https://link.aps.org/doi/10.1103/PhysRevB.93.184428}.

\bibitem[Mani and Karimipour(2015)]{mani2015}
Azam Mani and Vahid Karimipour.
\newblock Cohering and decohering power of quantum channels.
\newblock \emph{Phys. Rev. A}, 92:\penalty0 032331, Sep 2015.
\newblock \doi{10.1103/PhysRevA.92.032331}.
\newblock URL \url{https://link.aps.org/doi/10.1103/PhysRevA.92.032331}.

\bibitem[Marvian and Spekkens(2014)]{marvian2014}
Iman Marvian and Robert~W Spekkens.
\newblock Extending noether's theorem by quantifying the asymmetry of quantum
  states.
\newblock \emph{Nature Communications}, 5:\penalty0 3821 EP --, 05 2014.
\newblock \doi{10.1038/ncomms4821}.

\bibitem[Matera et~al.(2016)Matera, Egloff, Killoran, and Plenio]{matera2016}
J~M Matera, D~Egloff, N~Killoran, and M~B Plenio.
\newblock Coherent control of quantum systems as a resource theory.
\newblock \emph{Quantum Science and Technology}, 1\penalty0 (1):\penalty0
  01LT01, 2016.
\newblock \doi{10.1088/2058-9565/1/1/01LT01}.

\bibitem[Misra et~al.(2016)Misra, Singh, Bhattacharya, and Pati]{misra2016}
Avijit Misra, Uttam Singh, Samyadeb Bhattacharya, and Arun~Kumar Pati.
\newblock Energy cost of creating quantum coherence.
\newblock \emph{Phys. Rev. A}, 93:\penalty0 052335, May 2016.
\newblock \doi{10.1103/PhysRevA.93.052335}.
\newblock URL \url{https://link.aps.org/doi/10.1103/PhysRevA.93.052335}.

\bibitem[Napoli et~al.(2016)Napoli, Bromley, Cianciaruso, Piani, Johnston, and
  Adesso]{napoli2016}
Carmine Napoli, Thomas~R. Bromley, Marco Cianciaruso, Marco Piani, Nathaniel
  Johnston, and Gerardo Adesso.
\newblock Robustness of coherence: An operational and observable measure of
  quantum coherence.
\newblock \emph{Phys. Rev. Lett.}, 116:\penalty0 150502, Apr 2016.
\newblock \doi{10.1103/PhysRevLett.116.150502}.
\newblock URL \url{https://link.aps.org/doi/10.1103/PhysRevLett.116.150502}.

\bibitem[Narasimhachar and Gour(2015)]{narasimhachar2015}
Varun Narasimhachar and Gilad Gour.
\newblock Low-temperature thermodynamics with quantum coherence.
\newblock \emph{Nature Communications}, 6:\penalty0 7689 EP --, 07 2015.
\newblock \doi{10.1038/ncomms8689}.

\bibitem[Plenio and Virmani(2007)]{plenio2007}
Martin~B. Plenio and Shashank Virmani.
\newblock An introduction to entanglement measures.
\newblock \emph{Quantum Info. Comput.}, 7\penalty0 (1):\penalty0 1--51, January
  2007.
\newblock ISSN 1533-7146.
\newblock \doi{10.26421/QIC7.1-2}.

\bibitem[Sacchi(2005)]{sacchi2005}
Massimiliano~F. Sacchi.
\newblock Optimal discrimination of quantum operations.
\newblock \emph{Phys. Rev. A}, 71:\penalty0 062340, Jun 2005.
\newblock \doi{10.1103/PhysRevA.71.062340}.
\newblock URL \url{https://link.aps.org/doi/10.1103/PhysRevA.71.062340}.

\bibitem[Shao et~al.(2015)Shao, Xi, Fan, and Li]{shao2015}
Lian~He Shao, Zhengjun Xi, Heng Fan, and Yongming Li.
\newblock {Fidelity and trace-norm distances for quantifying coherence}.
\newblock \emph{Phys. Rev. A}, 91\penalty0 (4), 2015.
\newblock ISSN 10941622.
\newblock \doi{10.1103/PhysRevA.91.042120}.
\newblock URL
  \url{https://journals.aps.org/pra/pdf/10.1103/PhysRevA.91.042120}.

\bibitem[Stahlke and Griffiths(2011)]{Stahlke2011}
Dan Stahlke and Robert~B. Griffiths.
\newblock Entanglement requirements for implementing bipartite unitary
  operations.
\newblock \emph{Phys. Rev. A}, 84:\penalty0 032316, Sep 2011.
\newblock \doi{10.1103/PhysRevA.84.032316}.
\newblock URL \url{https://link.aps.org/doi/10.1103/PhysRevA.84.032316}.

\bibitem[Streltsov et~al.(2015)Streltsov, Singh, Dhar, Bera, and
  Adesso]{streltsov2015}
Alexander Streltsov, Uttam Singh, Himadri~Shekhar Dhar, Manabendra~Nath Bera,
  and Gerardo Adesso.
\newblock {Measuring Quantum Coherence with Entanglement}.
\newblock \emph{Physical Review Letters}, 115\penalty0 (2), feb 2015.
\newblock ISSN 10797114.
\newblock \doi{10.1103/PhysRevLett.115.020403}.
\newblock URL \url{https://arxiv.org/abs/1502.05876}.

\bibitem[Streltsov et~al.(2017)Streltsov, Adesso, and Plenio]{streltsov2017}
Alexander Streltsov, Gerardo Adesso, and Martin~B. Plenio.
\newblock Colloquium: Quantum coherence as a resource.
\newblock \emph{Rev. Mod. Phys.}, 89:\penalty0 041003, Oct 2017.
\newblock \doi{10.1103/RevModPhys.89.041003}.
\newblock URL \url{https://link.aps.org/doi/10.1103/RevModPhys.89.041003}.

\bibitem[Tóth and Apellaniz(2014)]{toth2014}
Géza Tóth and Iagoba Apellaniz.
\newblock Quantum metrology from a quantum information science perspective.
\newblock \emph{Journal of Physics A: Mathematical and Theoretical},
  47\penalty0 (42):\penalty0 424006, 2014.
\newblock \doi{10.1088/1751-8113/47/42/424006}.

\bibitem[Vedral and Plenio(1998)]{vedral1998}
V.~Vedral and M.~B. Plenio.
\newblock Entanglement measures and purification procedures.
\newblock \emph{Phys. Rev. A}, 57:\penalty0 1619--1633, Mar 1998.
\newblock \doi{10.1103/PhysRevA.57.1619}.
\newblock URL \url{https://link.aps.org/doi/10.1103/PhysRevA.57.1619}.

\bibitem[Watrous(2009)]{watrous2009}
John Watrous.
\newblock Semidefinite {P}rograms for {C}ompletely {B}ounded {N}orms.
\newblock \emph{Theory of Computing}, 5:\penalty0 217--238, 2009.
\newblock \doi{10.4086/toc.2009.v005a011}.

\bibitem[Winter and Yang(2016)]{winter2016}
Andreas Winter and Dong Yang.
\newblock Operational resource theory of coherence.
\newblock \emph{Phys. Rev. Lett.}, 116:\penalty0 120404, Mar 2016.
\newblock \doi{10.1103/PhysRevLett.116.120404}.
\newblock URL \url{https://link.aps.org/doi/10.1103/PhysRevLett.116.120404}.

\bibitem[Wootters and Zurek(1982)]{wootters1982}
W.~K. Wootters and W.~H. Zurek.
\newblock A single quantum cannot be cloned.
\newblock \emph{Nature}, 299:\penalty0 802 EP --, 10 1982.
\newblock \doi{10.1038/299802a0}.

\bibitem[Zhao et~al.(2018)Zhao, Liu, Yuan, Chitambar, and Ma]{zhao2018}
Qi~Zhao, Yunchao Liu, Xiao Yuan, Eric Chitambar, and Xiongfeng Ma.
\newblock One-shot coherence dilution.
\newblock \emph{Phys. Rev. Lett.}, 120:\penalty0 070403, Feb 2018.
\newblock \doi{10.1103/PhysRevLett.120.070403}.
\newblock URL \url{https://link.aps.org/doi/10.1103/PhysRevLett.120.070403}.

\bibitem[Zhu et~al.(2017)Zhu, Hayashi, and Chen]{zhu2017}
Huangjun Zhu, Masahito Hayashi, and Lin Chen.
\newblock Coherence and entanglement measures based on rényi relative
  entropies.
\newblock \emph{Journal of Physics A: Mathematical and Theoretical},
  50\penalty0 (47):\penalty0 475303, 2017.
\newblock \doi{10.1088/1751-8121/aa8ffc}.

\bibitem[Çakmak et~al.(2015)Çakmak, Karpat, and Fanchini]{cakmak2015}
Barış Çakmak, Göktuğ Karpat, and Felipe~F. Fanchini.
\newblock Factorization and criticality in the anisotropic xy chain via
  correlations.
\newblock \emph{Entropy}, 17\penalty0 (2):\penalty0 790--817, 2015.
\newblock ISSN 1099-4300.
\newblock \doi{10.3390/e17020790}.
\newblock URL \url{http://www.mdpi.com/1099-4300/17/2/790}.

\end{thebibliography}

%%%%%%%%%%%%%%%%%%%%%%%%%%%%%%%%%%%%%%%%%%%%%%%%%%%%%%%%%%%%%%%%%%%%%%%%%%%%%%%%%%%%%%%

%%%%%%%%%%%%%%%%%%%%%%%%%%%%%%%%%%%%%%%%%%%%%%%%%%%%%%%%%%%%%%%%%%%%%%%%%%%%%%%%%%%%%%%

\onecolumn\newpage
\appendix

\section{Appendix}
\subsection{Resource theory of coherence: main definitions}
\label{app:RoC}
In this appendix we provide a brief review of the resource theory of coherence focusing mainly on those elements of the theory 
that are relevant to our work.  For a more exhaustive review of the resource theory of coherence we direct the reader 
to~\cite{streltsov2017,hu2017}.

The set of \emph{free states}, $\Delta$, in the resource theory of coherence, also referred to as \emph{incoherent} states, form a 
convex subset of the set of all density matrices, $\cal{S(H)}$, and are defined as those density matrices that are diagonal with 
respect to a chosen orthonormal basis $\{\ket{j}:j=0,\ldots,d-1\}$ of the Hilbert space $\cal H$, 
$\Delta=\left\{\delta: \delta=\sum_{j=0}^{d-1}\delta_j\ket{j}\!\bra{j}\right\}$.  Free operations are those that map the set of free states 
onto itself; the largest class of free operations conceivable in a resource theory of coherence, and the primary focus of our work, 
are the so-called \emph{maximally incoherent operations (MIOs)}~\cite{aberg2006}, and they are simply defined as all CPTP 
operations, $\cM$, that map $\Delta$ onto itself. In particular, the Choi-Jamio\l kowski matrix, $J_{\cM}$, of a MIO operation 
$\cM:A\rightarrow B$ is characterized by the following conditions, in addition to the standard ones:
\begin{equation}
  \tr((\ket{i}\!\bra{i}\otimes \ket{j}\!\bra{j'})J_{\cM})=0, \ \forall i \ \forall \ j\neq j',
\label{app:MIO_condition}
\end{equation}
which is equivalent to requiring that $\cM$ does not generate coherence from incoherent states.  A well studied subset of MIO are 
the \emph{incoherent operations} (IOs), which are defined as those CPTP operations admitting a Kraus decomposition
$\cM(\rho)=\sum_{\alpha} K_\alpha\,\rho K_\alpha^\dag$, where each Kraus operator preserves incoherence, i.e., 
$K_\alpha \Delta K_\alpha^\dag \subset\Delta$. 

Any state that possesses coherence with respect to the chosen orthonormal basis qualifies as a resource.  The quantification of 
the resourcefulness of any given quantum state is accomplished by a suitable \emph{coherence measure}, 
$C:{\cal S(H)}\to\mathbb{R}_{\geq 0}$.  A good coherence measure should satisfy the following two conditions:
(i) $C(\rho)=0$ for all $\rho\in\Delta$, and (ii) $C(\rho)\geq C(\cT(\rho))$ for all incoherent CPTP maps $\cT$.  Another convenient---
but not necessary---condition is convexity, $\sum_i p_i C(\rho_i)\geq C(\sum_i p_i \rho_i)$.  The following are examples of 
coherence measures fulfilling all the previous conditions under MIO.

\begin{enumerate}
\item \emph{Robustness of coherence~\cite{napoli2016}}: quantifies the minimum amount of another state $\rho'$ required such 
that its convex sum with $\rho$ is an incoherent state. It can be cast as a semidefinite program in primal standard form
\begin{equation}
 1+C_R(\rho)=\min \{\lambda: \ \rho \leq \lambda \sigma,\,\sigma\in\Delta\}. 
\label{app:robStPrim}
\end{equation}
Its equivalent dual form is given by
\begin{equation}
1+C_R(\rho) = \max\{ \tr \rho S: S\geq 0,\, S_{jj}=1 \,\forall j\},
\label{app:robStDual}
\end{equation}
which holds because strong duality is fulfilled. The robustness is multiplicative under tensor product of states~\cite{zhu2017}:
\[
	\bfrob{\rho_1\otimes\rho_2}=(\bfrob{\rho_1})(\bfrob{\rho_2}).
\]
For the sake of completeness, we provide a self-contained proof of this fact below, see Lemma~\ref{lemma:multipRobSt}.  It is convenient to define a logged version of the 
robustness of coherence as:
\begin{equation}
C_{LR}=\log(1+C_R(\rho)),
\label{app:logged_roc}
\end{equation}
which now becomes an additive quantity under tensor product.

\item \emph{Relative entropy of coherence~\cite{baumgratz2014}}: 
\begin{equation}
C_r(\rho)=S(\cal{D}(\rho))-S(\rho),
\label{app:rec}  
\end{equation}
where $S(\rho)=-\tr(\rho \log \rho)$ is the von Neumann entropy, and $\cal{D}:\cal{S(H)}\to\cal{S(H)}$ is the map that erases all the off-diagonal elements of $\rho\in\mathcal{S(H)}$.
	
\item The \emph{$\ell_1$-norm of coherence~\cite{baumgratz2014}},
\begin{equation} 
C_{\ell_1}(\rho)=\sum_{i\neq j}|\rho_{ij}|,
\label{app:l1nc}
\end{equation}
is a coherence measure under IOs, but not MIOs~\cite{bu2017b}. It is popular because it is simple to compute. More importantly, 
the $\ell_1$-norm tightly bounds the robustness of coherence~\cite{napoli2016}
\begin{equation}
\frac{C_{\ell_1}(\rho)}{d-1} \leq C_R(\rho) \leq C_{\ell_1}(\rho),
\label{app:l1nc_bound_roc}
\end{equation}
where the upper bound becomes an equality for qubits and any pure state. We also note that $\log(1+C_{\ell_1}(\rho))$ is additive
under tensor product~\cite{zhu2017}.
	
\item \emph{Coherence rank and coherence number~\cite{killoran16}}:  For pure states the coherence rank is defined as
\begin{equation}
C_{\mathrm{rank}}(\psi):= \min \left\{r : \ket{\psi}=\sum_{i=1}^{r}c_{i}\ket{c_i} \right\},
\label{app:cohrank}
\end{equation}
where $\ket{c_i}$ is an element of the incoherent basis. That is, the coherence rank of $\ket{\psi}$ is the number of non-zero terms 
that appear when writing  $\ket{\psi}$ in the incoherent basis. This coherence measure can be extended by a convex roof 
construction, to the \emph{coherence number}: 
\begin{equation}
 C_{\mathrm{rank}}(\rho):=\hspace{-2mm}\min_{\rho = \sum_\alpha p_\alpha \psi_\alpha} \max_{\alpha} \,C_{\mathrm{rank}}
 (\psi_\alpha).
 \label{app:cohnumber}
 \end{equation}
We stress that, like $C_{\ell_1}$, $C_{\mathrm{rank}}$ is a monotone only under IOs, not under MIOs.
\end{enumerate}

\begin{lemma}[Zhu \emph{et al.} \cite{zhu2017}]
\label{lemma:multipRobSt}
For any tensor-product state on a bipartite system, it holds 
\begin{equation}
\bfrob{\rho_1\otimes\rho_2}=(\bfrob{\rho_1})(\bfrob{\rho_2}).
\end{equation} 
\end{lemma}
\begin{proof}
Consider the two values $\lambda_i=\bfrob{\rho_i}$ for $i=1,2$. Then by 
\eqq{app:robStPrim} there exist $\sigma_i\in\Delta$ such that $\rho_i\leq\lambda_i\sigma_i$, which implies 
$\rho_1\otimes\rho_2\leq\lambda_1\lambda_2 \;\sigma_1\otimes\sigma_2$, i.e., that 
$\bfrob{\rho_1\otimes\rho_2}\leq\lambda_1\lambda_2$. 

On the other hand, take $S_i$ that optimize \eqq{app:robStDual} for both local states, i.e., $\bfrob{\rho_i}=\tr(\rho_i S_i)$ for 
$i=1,2$, and take as an ansatz for the optimizer of the total robustness $S=S_1\otimes S_2$.  It follows that 
$\bfrob{\rho_1\otimes\rho_2}\geq\tr(\rho_1 S_1)\tr(\rho_2 S_2)$. 
\end{proof}

%%%%%%%%%%%%%%%%%%%%%%%%%%%%%%%%%%%%%%%%%%%%%%%%%%%%%%%%%%%%%%%
\subsection{The robustness of coherence of a channel}
\label{app:RobCoh}
In this subsection we prove several properties of the robustness of coherence for quantum channels, introduced in 
\eqq{robCohChDef}. Let us start by showing that it is a coherence measure for a coherence theory whose free resources are MIOs.

\begin{enumerate}
\item \emph{Faithfulness}: $C_R(\cN)=0$ if and only if $\cN\text{ is MIO}$. This follows straightforwardly from the definition, 
since one can take $\lambda=1$ if and only if $\cM\equiv \cN$.

\item \emph{Monotonicity}: The robustness of coherence of the channel $\cN$ monotonically decreases if we pre- and 
post-process the input and output using MIO channels $\cL'$ and $\cL$ respectively, i.e., 
$C_R(\cL\circ \cN \circ \cL' )\leq C_R(\cN),\; \forall\,\cL,\,\cL'\text{ is MIO}$. Indeed, take 
$\lambda\geq0, \cM \text{ MIO}$ such that $\lambda \cM-\cN\geq0$ and concatenate it with the channels $\cL, \cL'$. Defining $\cM'=\cL\circ \cM\circ \cL'\text{ MIO}$ it follows that 
$\cL\circ \cN\circ \cL'\leq\lambda \cM'$ as well, so that any $\lambda$ feasible for $\cN$ is also feasible for 
$\cL\circ \cN\circ \cL'$.

\item \emph{Convexity}: $\sum_i p_i C_R(\cN_i)\geq C_R(\sum_i p_i \cN_i)$ for any probability distribution $\{p_i\}$ and 
$\{\cN_i\}$ a collection of quantum channels. Indeed for each $i$ take the minimum $\lambda_i$ such that 
$\cN_i\leq\lambda_i \cM_i,\, \cM_i \mathrm{MIO}$. Then define $\bar{\lambda}=\sum_ip_i\lambda_i$ and note that 
$\{\widetilde{p}_i=p_i\lambda_i/\bar{\lambda}\}$ is still a probability distribution. By averaging the inequalities over $\{p_i\}$ and 
rescaling by $\bar{\lambda}$ we get $\sum_i p_i \cN_i\leq \bar{\lambda} \sum_i\widetilde{p}_i\cM_i$, where 
the latter is still a MIO. We conclude that $1+C_R(\sum_i p_i \cN_i)\leq \bar{\lambda}=\sum_i\lambda_i\,p_i$, and thus 
$\sum_i p_i C_R(\cN_i)\geq C_R(\sum_i p_i \cN_i)$.
\end{enumerate}

Properties 1-3 above straightforwardly extend to the smooth-robustness of \eqq{smoothRob} (the faithfulness condition holds up to error $\epsilon$).

Let us now focus on the SDP formulation of the robustness.

\begin{proof-of}{Theorem~\ref{theorem_dual}}
We want to calculate the dual of the SDP
\begin{equation*}
  1+C_R(\cN) = \min \{\lambda: \cM' \text{ is MIO and } \cN\leq \lambda \cM' =: \cM\},
\end{equation*}
which is equivalent to 
\begin{align}
\min &\ \lambda\\
   \text{\rm s.t.}&\ J_{\cN} \leq J_{\cM} \nonumber\\
   &\tr_B J_{\cM} =\lambda \mathds{1}_A\nonumber \\ 
  & \tr J_{\cM}(\ket{i}\!\bra{i}\otimes \ket{j}\!\bra{k}) =0,\ \forall i,\,\forall j\neq k, \nonumber
\end{align}
where $J_{\cN}$, $J_{\cM}$ are the Choi matrices of $\cN$ and $\cM$, respectively. This means that we need to minimize the following Lagrangian:	
\begin{eqnarray}
	 L&=&\lambda +\tr X (\tr_B J_{\cM} -\lambda \mathds{1}_A)-\tr Y(J_{\cM}-J_{\cN})\nonumber\\
	&&+\sum_{i,j\neq k}z_{jk}^i \tr J_{\cM}(\ket{i}\bra{i}\otimes \ket{j}\bra{k})\nonumber\\
	&=& \lambda +\tr(X\otimes \mathds{1})J_{\cM} -\lambda \tr X -\tr Y(J_{\cM}-J_{\cN})\nonumber\\
	&&+\tr J_{\cM} \Big(\sum _{i,j\neq k}\ket{i}\bra{i}\otimes Z_i\Big)\\
	&=& \lambda +\tr(X\otimes \mathds{1})J_{\cM} -\lambda \tr X -\tr Y(J_{\cM}-J_{\cN})\nonumber\\
	&&+\tr (J_{\cM} Z)\nonumber\\
	&=&\lambda(1-\tr X)+\tr J_{\cM}(X\otimes \mathds{1}-Y+Z)+\tr(J_{\cN} Y),\nonumber
\end{eqnarray}
where $Z=\sum_i \dketbra{i} \otimes Z_i$ and $Z_i$ is a zero-diagonal matrix, while $Y\geq 0.$ In order to avoid that $\min_{\lambda,J_{\cM}} L=-\infty$, we need to impose that $\tr X=1$ (so $X$ is a state) and $Y=X\otimes \mathds{1}+Z$.

To calculate the dual, we maximize the terms of the Lagrangian which 
do not contain the variables of the primal form, $\lambda$ and $J_{\cM}$:
\begin{align}
 \max \{\tr J_{\cN}Y: \tr X=1, Y=X\otimes \mathds{1}+Z\geq 0\}.
 \end{align}
 The objective function is equivalent to
 \begin{align}
	\tr J_{\cN}Y = &  \tr J_{\cN}(X\otimes \mathds{1})+\tr(J_{\cN} Z)\nonumber\\
	=& 1+\tr(J_{\cN} Z)\nonumber\\
	= & 1+\sum_i \tr Z_i \cN(\ket{i}\!\bra{i}),
\label{app:tmp1}
\end{align}
where we have made use of the identity $\tr_B((X\otimes \1)J_{\cN})=\cN(X^T)$, and the fact that $\cN$ is CPTP in the second 
line of Eq.~\eqref{app:tmp1}. As $X$ no longer appears in the objective function, and dephasing it does not affect the 
constraints, we may assume $X=\sum_i p_i \ket{i}\!\bra{i}$, where $p_i$ are probabilities, without loss of generality. 
Therefore, we can write
\begin{align}
  & \max\Big\{\sum_i p_i + \sum_i \tr Z_i \cN(\ket{i}\!\bra{i}): p_i\mathds{1}+Z_i\geq 0\, \forall i\Big\} \nonumber\\
	= & \max \sum_i p_i \tr S_i \cN(\ket{i}\!\bra{i})\nonumber\\
	= &\max \sum_i p_i (1+C_R(\cN(\ket{i}\!\bra{i})))\nonumber\\ 
	= &\max_i \{1+C_R(\cN(\ket{i}\!\bra{i}))\},
\label{app:roc_dual}
\end{align}
where $S_i=\1+p_i^{-1}Z_i\geq 0$ and $\bra{j}S_i\ket{j}=1$ for all $j$, and we have used the dual form of robustness for states, 
\eqq{app:robStDual}.  Moreover, since the latter is multiplicative under tensor product, see Lemma \ref{lemma:multipRobSt} in 
Sec.~\ref{app:RoC}, also the robustness of channels is.

The logged version of the robustness of coherence follows directly from Eq.~\eqref{app:roc_dual}, and the equality of the 
log-robustness of coherence to the cohering power, $P_R(\cN)$, of Eq.~\eqref{robCohPow} is also evident from 
the definition of the latter.  To see that the log-robustness of coherence for a channel is also equal to $\widehat{P}(\cN)$, observe that for any $\rho\otimes\omega$
\begin{align}\nonumber
&C_{LR}(\rho\otimes\omega)\geq C_{LR}(\cM(\rho\otimes\omega))=C_{LR}(\cN(\rho)\otimes\sigma)\\ \nonumber
&C_{LR}(\rho)+C_{LR}(\omega)\geq C_{LR}(\cN(\rho))+C_{LR}(\sigma)\\
&C_{LR}(\omega)-C_{LR}(\sigma)\geq C_{LR}(\cN(\rho))-C_{LR}(\rho).
\label{app:coherence_in_coherence_out}
\end{align}
Thanks to Corollary~\ref{corollary2}, the minimization of the right-hand side of Eq.~\eqref{app:coherence_in_coherence_out} over $\omega$ and $\sigma$ 
yields the amortized cost (Eq.~\eqref{eq_amortized}), which is equal to the log-robustness of coherence of
$\cN$ via Theorem~\ref{amortizedc}.  Moreover, the inequality in Eq.~\eqref{app:coherence_in_coherence_out} holds for all 
$\rho\in\cal{H}$.  Thus
\begin{equation}
C_{LR}(\cN)\geq \max_{\rho\in\cal{S(H)}}\left(C_{LR}(\cN(\rho))-C_{LR}(\rho)\right)\geq P_R(\cN),
\label{app:lrocequalcohpower}
\end{equation} 
where the last inequality holds because the maximization in $\widehat{P}_R(\cN)$ is over a larger convex set than that of 
$P_R(\cN)$.  As the upper and lower bounds on the coherence power are equal it follows that 
$C_{LR}(\cN)=\widehat{P}_R(\cN)$.  This completes the proof.
\end{proof-of}

Similarly, one can write the smoothed robustness of coherence (Definition~\ref{RobDef}) as an SDP in primal and dual form.  
Using the dual-SDP formulation of the diamond norm~\cite{watrous2009}, the primal SDP of  \eqq{smoothRob} reads
\begin{align} \nonumber
1+C_R^{\epsilon}(\cN) = \min &\ \lambda\\   \nonumber
\text{\rm s.t.}& \ J_{\cM}\geq J_{\cL},\\  \nonumber
& \tr_{B} J_{\cM}=\lambda\mathds{1}_{A},\\     \nonumber
& \tr{J_{\cM}(\ketbra{ij}{ik})}=0, \ \forall i\,\forall j\neq k,\\    \nonumber
& V\geq J_{\mathcal{L}}-J_{\cN},\\   \nonumber
& \tr_{B} V \leq \epsilon\mathds{1}_{A}, \\ \nonumber
& \tr_{B} J_{L} =\mathds{1}_{A},\\  
& J_{L}\geq0,\, V\geq 0.
\label{app:smoothedROC_primalSDP}
\end{align}
The first three constraints correspond to the simulation of channel $\cL$ by MIO and the fourth and fifth 
constraints capture the diamond norm constraints that $\cL$ should be $\epsilon$ close to $\cN$.  The dual form 
of Eq.~\eqref{app:smoothedROC_primalSDP} is given by 
\begin{align}\nonumber
\max &\,\tr(J_{\cN}(S-\Delta-\epsilon\Omega\otimes\1))\\    \nonumber
\text{\rm s.t.}&\  S\equiv Y\otimes\1-\sum_{i}\ketbra{i}{i}\otimes Z_i\geq 0\\   \nonumber
&S-\Delta\leq W\otimes\1\leq S-\Delta+\Omega\otimes \1\\    \nonumber
&\tr Z_i=0,\, \forall i\\   \nonumber
&\tr W=0,\,\tr Y=1\\
&\Delta\geq 0,\, \Omega\geq 0,\, Y\geq 0.
\label{app:smoothedROC_dualSDP}
\end{align}

Finally let us state the connection between the log-robusteness of coherence of a channel (correspondingly its smoothed 
version) to the maximum relative entropy (correspondingly $\epsilon$-maximum relative entropy) of the latter with respect to a 
MIO:
\begin{definition}	\label{app:prop_2}
	The maximum relative entropy of two channels $\cN$ and $\cM$ is:
	\begin{align}\label{app:maxRelEnt}\nonumber
	D_{\max}(\cN\|\cM) := -\log \max &\ p \\   \nonumber
 \text{\rm s.t.,}& \ \cM=p\cN+(1-p)\cN'\nonumber\\
	   & p\in[0,1],\ \cN' \in \text{CPTP},\nonumber\\
	   & \hspace{-1.95cm}= \log \min\left\{\lambda: \cN\leq \lambda \cM\right\} .
	\end{align}
The smoothed version of this quantity is:
\begin{equation}\label{app:epsilon_max-rel_entropy}
	D_{\max}^\epsilon(\cN\|\cM):= \min \left\{D_{\max}(\cL\|\cM): \frac{1}{2}\|\cN-\cL\|_\diamond\leq \epsilon\right\}.
	\end{equation}
\end{definition}

\begin{proposition}
	The $\epsilon$-log-robustness of a channel is  given by
	\begin{align}
	C_{LR}^\epsilon(\cN)%:&=\log(1+C_R^\epsilon(\cN))\nonumber\\
	=\min\left\{ D_{\max}^\epsilon(\cN\|\cM): \cM \text{MIO}\right\}.
	\label{app:logged-epsilon-robust}
	\end{align}
	\label{app:prop_3}
\end{proposition}

\subsection{Channel implementation with maximally coherent resources}
\label{app:cosdit_implementation}
In this subsection we prove the main results related to channel implementation with cosdits, Sec.~\ref{cosdits}.
We begin with the implementation cost without recycling:
 
\begin{proof-of}{Theorem \ref{theorem1}}
Let $\cL:=\cM(\Psi_k \otimes \cdot):A\rightarrow B$ be the MIO simulation of the 
induced channel $\cL:A\rightarrow B$ such that $\frac{1}{2}\|\cN-\cL\|_\diamond\leq\epsilon$. 
Noting that $\Psi_k + \left(k-1\right)\sigma = \mathds{1}$, with $\sigma = \dfrac{\mathds{1}-\Psi_k}{k-1}$, define 
\begin{align}\label{M_0} \nonumber
\cM' := \cM\left(\dfrac{\mathds{1}}{k} \otimes \cdot\right)&= \frac{1}{k}\cM\left((\Psi_k+(k-1)\sigma)\otimes \cdot\right)\nonumber\\
&=\frac{1}{k}\cM(\Psi_k\otimes \cdot)+\left(1-\frac{1}{k}\right)\cM(\sigma\otimes\cdot)\nonumber\\
&=\frac{1}{k}\cL+\left(1-\frac{1}{k}\right)\cL',
\end{align}                     
with $\cL' = \cM(\sigma \otimes \cdot)$. As $\cM$ is MIO, so is $\cM':A\rightarrow B$, and the 
right-hand-side of Eq.~(\ref{M_0}) corresponds to a convex decomposition of $\cM'$ in terms of  $\cL,\cL'\in\mathrm{CPTP}$. 
Hence, from the definition of the maximum-relative entropy, Eq.~\eqref{app:maxRelEnt} it follows
\begin{equation}
\log k \geq D_{\max}^\epsilon(\cL\|\cM')\geq C_{LR}^\epsilon(\cN).
\end{equation}

For the converse, let $k \geq 1+C_R^\epsilon(\cN)$ be an integer. By Eq.~\eqref{app:logged-epsilon-robust}, it follows that
there exists some CPTP map $\cL$ with $\dfrac{1}{2}||\cN-\cL||_\diamond\leq \epsilon$
and another CPTP map $\cL'$, such that $(\cL + \left(k-1\right)\cL')/k\text{ is MIO}$. 
Make the following ansatz for a channel $\cM$ that is feasible for the simulation cost:
$$\cM(\tau \otimes \rho) := \tr(\Psi_k \tau) \cL(\rho) + \tr((\mathds{1}-\Psi_k)\tau) \cL'(\rho).$$
 The map $\cM$ is MIO if and
only if $\cM(\ket{i}\bra{i} \otimes \cdot)$ is MIO for all incoherent basis states $\ket{i}$.
This is the case, since $\cM(\ket{i}\bra{i} \otimes \cdot) = \dfrac{1}{k} \cL + \big(1-\dfrac{1}{k}\big)\cL'$, which is MIO by construction. 
Hence $\log k\geq C_{\simu}^{\epsilon}(\cN)$ and the former can be taken as small as $\lceil1+C_R^\epsilon(\cN)\rceil$.
As the implementation cost is necessarily an integer \eqq{eq_MIO-simulation} follows.  This completes the proof.
\end{proof-of}

The proof for the simulation cost with recycling relies on the previous result, effectively implementing a simulation of the 
target channel and the resource.

\begin{proof-of}{Theorem \ref{theoremk-sim}}
Let $\cT:\mathbb{C}\otimes A\rightarrow  S \otimes  B$ such that 
$\cT(1\otimes\rho)=\sigma\otimes\cN(\rho)$.  Here, we have made use of the fact that any state can be identified with a preparation 
channel from $\mathbb{C}$ to $ S$ mapping the unique state $1\to\sigma$.  
Now let $\cM \text{ be a MIO}$ such that 
$$\cM(\Psi_k\otimes\cdot)=\cT(\cdot).$$
From Theorem~\ref{theorem1} such a simulation is possible if and only if 
\begin{align*}
k&\geq (1+C_R(\cT))\\
&=(1+C_R(\sigma))(1+C_R(\cN))
\end{align*}
where we have used the multiplicativity of the robustness of coherence in the last line.  Taking the logarithm on both sides and 
re-arranging terms gives Eq.~\eqref{eq_coh_left}.  This completes the proof.
\end{proof-of}

We will now make use of Theorem~\ref{theoremk-sim} to prove Theorem~\ref{amortizedc} regarding the amortized cost.

\begin{proof-of}{Theorem \ref{amortizedc}}
From Definition~\ref{amorDef} and Corollary~\ref{corollary2} it follows immediately that the zero-error amortized cost is equal to the log-robustness. Indeed, if we take the log on each side of \eqq{interval} and then let $k,m\rightarrow\infty$ we obtain $\amocoz{\cN}=C_{LR}(\cN)$.
this in turn implies that the $\epsilon$-error amortized cost, \eqq{eq_amortized}, can be rewritten in terms of the log-robustness of the channel $\cL$ as
\begin{equation}
\amoco{\cN}=\min C_{LR}(\cL) \, \text{ s.t. }\frac{1}{2}||\cN-\cL||_\diamond\leq\epsilon\equiv C_{LR}(\cN).
\end{equation}
\end{proof-of}

\subsection{Channel implementation with arbitrary resources: SDPs}
\label{app:sdps}
Here we prove the SDP forms of the generic implementation problem considered in Sec.~\ref{noncosdits}.
We begin by first providing an SDP for the optimal implementation of unitary gates without recycling using the gate fidelity as a figure 
of merit.  We then proceed to prove Proposition~\ref{diamondnorm_sdp}. For ease of notation, we swap the 
output subspaces of the MIO implementation map throughout this section, i.e., $\cM:R\otimes A \rightarrow B\otimes S$. We also consider unnormalized Choi matrices.  
\begin{proposition}
  \label{app:proposition1}
The optimal gate fidelity of implementation of a unitary $U: A\rightarrow  B$ by means of a MIO 
$\cM: R\otimes A \rightarrow  B\otimes S$ and a pure coherent state 
$\omega\in R$ is given by the following SDP:
  \begin{equation}\begin{split}
	  \label{app:thesdp}
     F= \max &\ \frac{1}{d_A^2}\tr((\omega^T\otimes J_U\otimes \mathds{1}_{S})X) \\
           \text{\rm s.t.}\, &\ X \text{ is the Choi matrix of a MIO},
  \end{split}\end{equation}
  where $J_U$ is the Choi matrix of the channel $U\cdot U^\dagger$, $d_A=\mathrm{dim} (A)$, and $\omega^T$ denotes the transpose of $\omega$.
\end{proposition}
\begin{proof}
Consider a maximally incoherent operation $\cM: R\otimes A \rightarrow  B\otimes S$. 
Its Choi matrix is defined as 
\begin{align}\nonumber
  J_{\cM}&=(\id_{RA}\otimes \cM)(\Phi_{RR'} \otimes \Phi_{AA'})\\
  &=\sum_{i,k,j,l}\ket{ik}_{RA}\bra{jl}\otimes\cM\left(\ket{ik}_{R'A'}\bra{jl}\right).
  \label{app:MIO_Choi}
\end{align} 
Suppose that $\cE=\tr_S(\cM(\omega \otimes \cdot))$, where $\cE: A \rightarrow B$.  Then the Choi matrix of $\cE$, $J_\cE$, is explicitly given by 
\begin{align}\nonumber
J_{\cE} &= \tr_{RS}((\omega^T \otimes \1_{A} \otimes \1_{B} \otimes \1_{S}) J_{\cM})\\  \nonumber
&=\sum_{ijkl}\bra{j}\omega^T\ket{i}\ket{k}_A\bra{l}\otimes\tr_S\left(\cM(\ket{ik}_{R'A'}\bra{jl})\right)\\  \nonumber
&=\sum_{k,l}\ket{k}_A\bra{l}\otimes\tr_S\left(\cM(\omega\otimes\ket{k}_{A'}\bra{l})\right)\\
&=(\id_A\otimes\cE)(\Phi_{AA'})
\label{app:CJofE}
\end{align}
On the other hand, the Choi matrix of the unitary channel is given by $J_U = (\1_A\otimes U)\Phi_{AA'}(\1_A\otimes U^\dagger)$.
 	
Our aim is to compute how well the map $\cE$ implements the unitary channel $\cN(\rho)=U\rho U^\dagger$. 
To that end, we use the \emph{gate fidelity}, i.e., the fidelity between the Choi matrices of the corresponding channels:
\begin{eqnarray}
 	F&=&\frac{1}{d_A^2}\tr(J_U J_{\cE}) \nonumber \\
 	&=& \frac{1}{d_A^2}\tr((\mathds{1}_{R} \otimes J_U \otimes \mathds{1}_{S}) (\omega^T \otimes \mathds{1}_{A}\otimes \mathds{1}_{B}\otimes \mathds{1}_{S})J_{\cM}) \nonumber \\
 	&=& \frac{1}{d_A^2}\tr((\omega^T\otimes J_U\otimes \mathds{1}_{S})J_{\cM}).
\end{eqnarray}
In particular, we want to obtain the optimal gate fidelity of implementation of the given channel.  This is an SDP in primal standard 
form, where $J_{\cM}\geq 0$ is the semidefinite variable subject to the constraints 
$\tr_{BS} J_{\cM} = \mathds{1}_{R}\otimes \mathds{1}_{A}$ and 
\mbox{$\tr((\ket{i k}_{RA}\bra{i k}\otimes \ket{j l}_{SB}\bra{j' l'})J_{\cM})=0$}
for all $i,k$ and all $j\neq j', l\neq l'$.
\end{proof}

In the case of a generic channel $\cN$, the diamond norm distance 
needs to be used, as discussed in the main text.

\begin{proof-of}{Proposition~\ref{diamondnorm_sdp}}
The SDP can now be formulated as follows:
\begin{equation}\begin{split}
  \label{diamondsdp}
  \min &\ \frac{1}{2}\|\cE-\cN\|_\diamond\\
  \text{\rm s.t.} &\ X  \text{ is the Choi matrix of a MIO } \cM, \\
		&\ \cE = \tr_B \cM(\omega\otimes\cdot).
\end{split}\end{equation}
Using the result of~\cite{watrous2009}, the dual form of the diamond norm distance can be written as  
\begin{align}\nonumber
\min &\ \lambda\\  \nonumber
\text{\rm s.t.} \ & Z\geq J_{\cE}-J_{\cN}\\  \nonumber
&  \tr_{B} Z\leq \lambda\1_A\\
&Z\geq0.
\label{app:diamond_dual}
\end{align}
Using Eq.~\eqref{app:CJofE} to relate the Choi matrix of $\cE$ to that of $\cM \text{ is MIO}$ it is then straightforward to obtain 
the final formulation.
\end{proof-of}

%%%%%%%%%%%%%%%%%%%%%%%%%%%%%%%%%%%%%%%%%%%%%%%%%%%%%%%%%%%%%%%%%%%%%%%%%%%%%%%%%%%%
\subsection{Alternative definitions of coherence left}
\label{app:altCohLeft}
In this appendix we discuss possible bounds on and alternative definitions of coherence left, when an arbitrary resource is used to simulate a channel with recycling. Let us start by relaxing the tensor-product constraint at the output of the implementation: a simpler problem is that of finding a MIO that simulates the whole tensor-product output up to a given error. This also constitutes an upper bound on the robustness of coherence left. 
 \begin{proposition}
   \label{cohLeftUpBound}
   The maximum robustness of coherence left in the resource $\sigma\in S$ after the 
   implementation of a quantum channel $\cN:A\rightarrow B$ via 
   MIO $\cM:R\otimes A \rightarrow S\otimes B$ and a coherent resource $\omega\in R$
   up to error $\epsilon$ in diamond norm, i.e., \eqq{cohLeftNonLinEQ}, is bounded from above by the following optimization problem:
   \begin{equation}\begin{split}
     \label{cohLeftUpBoundEQ}
     \max &\ C_{R}(\sigma)\\
      \text{\rm s.t. }& \frac{1}{2}||\sigma\otimes\cN-\cM(\omega\otimes\cdot)||_{\diamond}\leq \epsilon,
     \end{split}
     \end{equation}
which can be expressed as the following semidefinite program:
  \begin{equation}\begin{split}
     \label{cohLeftUpBoundEQSDP}
     \max &\ \sum_{i\neq j}\sigma_{ij}\\
     \text{\rm s.t.}&\ J_{\cM}  \text{ is the Choi matrix of } \cM \text{ MIO} \\
 	 &\ Z  \geq \sigma\otimes J_{\cN}-\tr_R( (\omega^T\otimes\mathds{1}_{SAB})J_{\cM})\\
 	 &\ \tr_{SB} Z\leq  \epsilon \mathds{1}_{A}\\
	 &\ Z  \geq 0.
     \end{split}
     \end{equation}
 	 where $J_\cN$ is the Choi matrix of $\cN$.
	  \end{proposition}

 \begin{proof}
 It is straightforward to show that \eqq{cohLeftUpBoundEQ} is an upper bound on \eqq{cohLeftNonLinEQ}. Indeed, we can make the separable-output ansatz $\cM(\omega\otimes\cdot)=\sigma\otimes\cL(\cdot)$ for the MIO $\cM$ simulating the channel in \eqq{cohLeftUpBoundEQ} and obtain the problem of \eqq{cohLeftNonLinEQ}.
 
 As for the SDP formulation of this problem, the last three conditions in \eqq{cohLeftUpBoundEQSDP} are a simple translation of the diamond norm error one in \eqq{cohLeftUpBoundEQ}, following~\cite{watrous2009}. 
  It then remains to show that the robustness of coherence can be re-cast as the sum of off-diagonal elements of 
 $\sigma\in S$.  To that end recall the dual formulation, Eq.~\eqref{app:robStDual}, of the robustness of 
 coherence which reads  
 \begin{equation}
 1+C_R(\sigma) = \max \tr \sigma S 
               = \max \tr (\sigma \circ S^T)G,
               \label{app:hadamard}
 \end{equation} 
where $S$ is a non-negative-definite matrix with ones on the diagonal and
we have introduced the all-one matrix $G_{ij}=1$ for all $i,j$ and the Hadamard component-wise product 
 $\sigma \circ S=\sum_{ij}\sigma_{ij}S_{ij}\ket{i}\bra{j}$. That this is equal to the sum of all off-diagonal elements of $\sigma$ follows 
 from
 \begin{align}
 C_R(\sigma)&=\sum_{i,j}\sigma_{ij}\,S_{ji}-1\nonumber\\
 &=\sum_{i}\sigma_{ii}+\sum_{i\neq j}\sigma_{ij}S_{ji}-1\nonumber\\
 &=\sum_{i\neq j}\sigma_{ij}S_{ji},
 \label{app:roceqoff-diagonal}
 \end{align} 
 where we have made use of the fact that $S_{ii}=1\,\forall\, i$ in the second line above.  As the transformation 
 $\sigma \rightarrow \sigma \circ S^T$ can be implemented by a MIO CPTP map, given that $(\sigma \circ S^T)_{ij}\leq\sigma_{ij} \; \forall i\neq j$ (no coherence is created in the transformation), the two objective functions 
 are equivalent under the given constraints, since the diamond norm is contractive under CPTP maps.
 \end{proof}	  
	  
The SDP of Proposition~\ref{cohLeftUpBound} amounts to requiring that the MIO $\cM$ implements a good simulation of the overall output $\sigma\otimes\cN$ when provided with the input resource $\omega$. Note that this implies that the local reduced output systems of the implementation $\cM$ are close to $\sigma$ and $\cN$, but also that their correlations are small, so that the implementation is $\epsilon$-close to a tensor-product one. However, if we increase the allowed dimension of $\sigma$ it is possible to find states that are $\epsilon$-close to it but have an increasing amount of coherence, e.g., $(1-\epsilon)\sigma+\epsilon\Psi_d$. Hence \eqq{cohLeftUpBoundEQ} can give an unbounded amount of coherence left if the output dimension of the resource is allowed to increase. 

Two possible ways to remedy this involve changing the objective function such that it  either maximizes the $\epsilon$-robustness of coherence of the output resource state, or to find the most coherent state among all output resource states that are $\epsilon$-close to the desired target.  Both of these solutions provide an alternative definition of coherence left, physically motivated by the fact that the true output resource is not $\sigma$ but only a state $\epsilon$-close to it, either in robustness or in trace norm.  Unfortunately, the former is not an SDP while the latter is not easily computable.   Which of these, or other, definitions of coherence left is better may depend on the recycler's objective and we leave it as an open question for the interested reader.

In Fig.~\ref{fig:coh_left_vs_uni} we plot the robustness of coherence left after the approximate implementation of a qubit unitary $U_\theta$ vs. $\theta/\pi$, for several values of the input coherence and of the error threshold. As expected, for a fixed input coherence, more coherence can be recovered at the output if we accept a worse implementation. Moreover, for sufficiently small error thresholds, there are unitaries that cannot be implemented at all. 

%\onecolumngrid
	  
 \begin{figure*}[t!]
 
 	\subfloat[]{ \includegraphics[width = 0.5\linewidth]{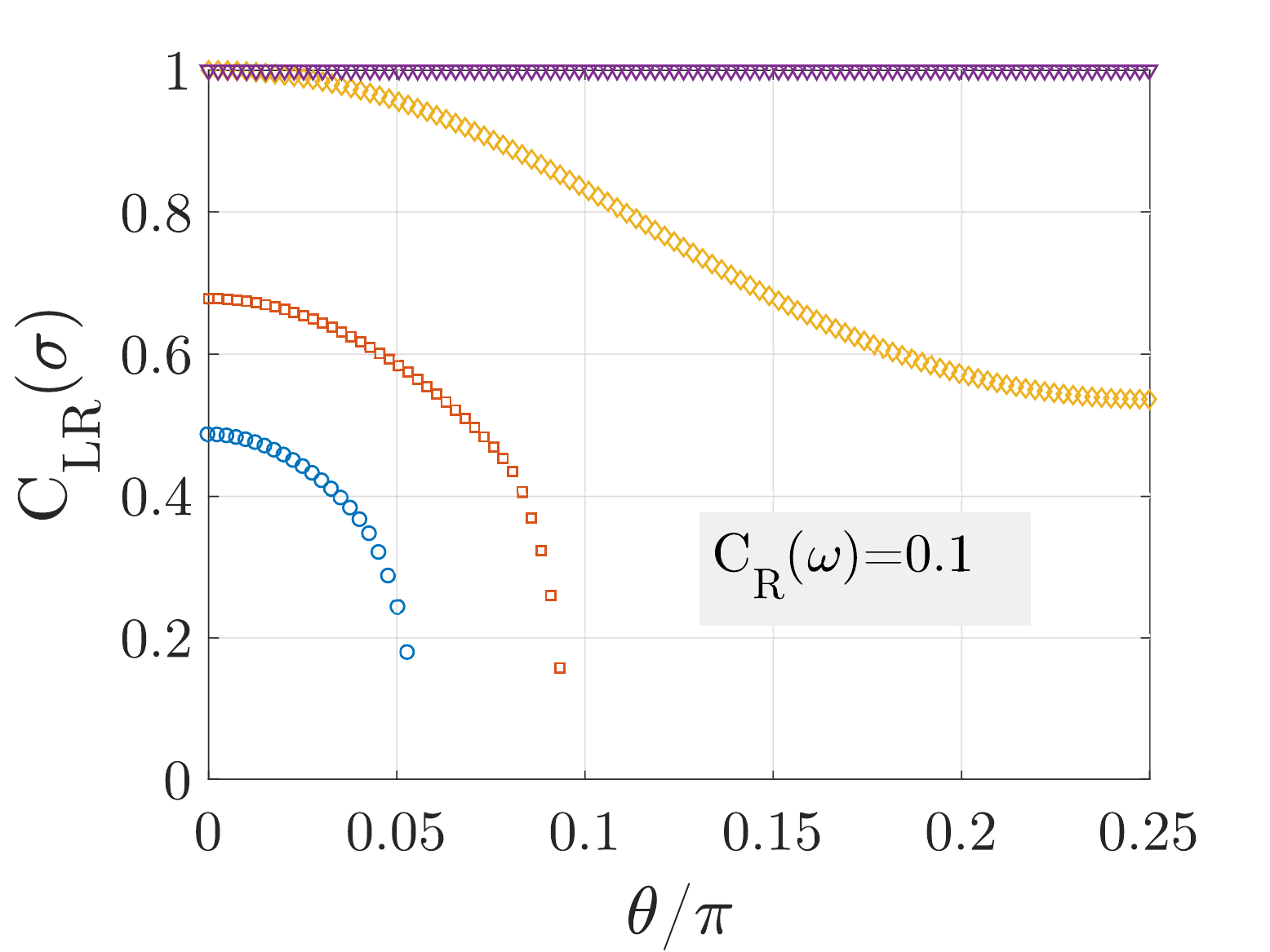}}
	\subfloat[]{ \includegraphics[width = 0.5\linewidth]{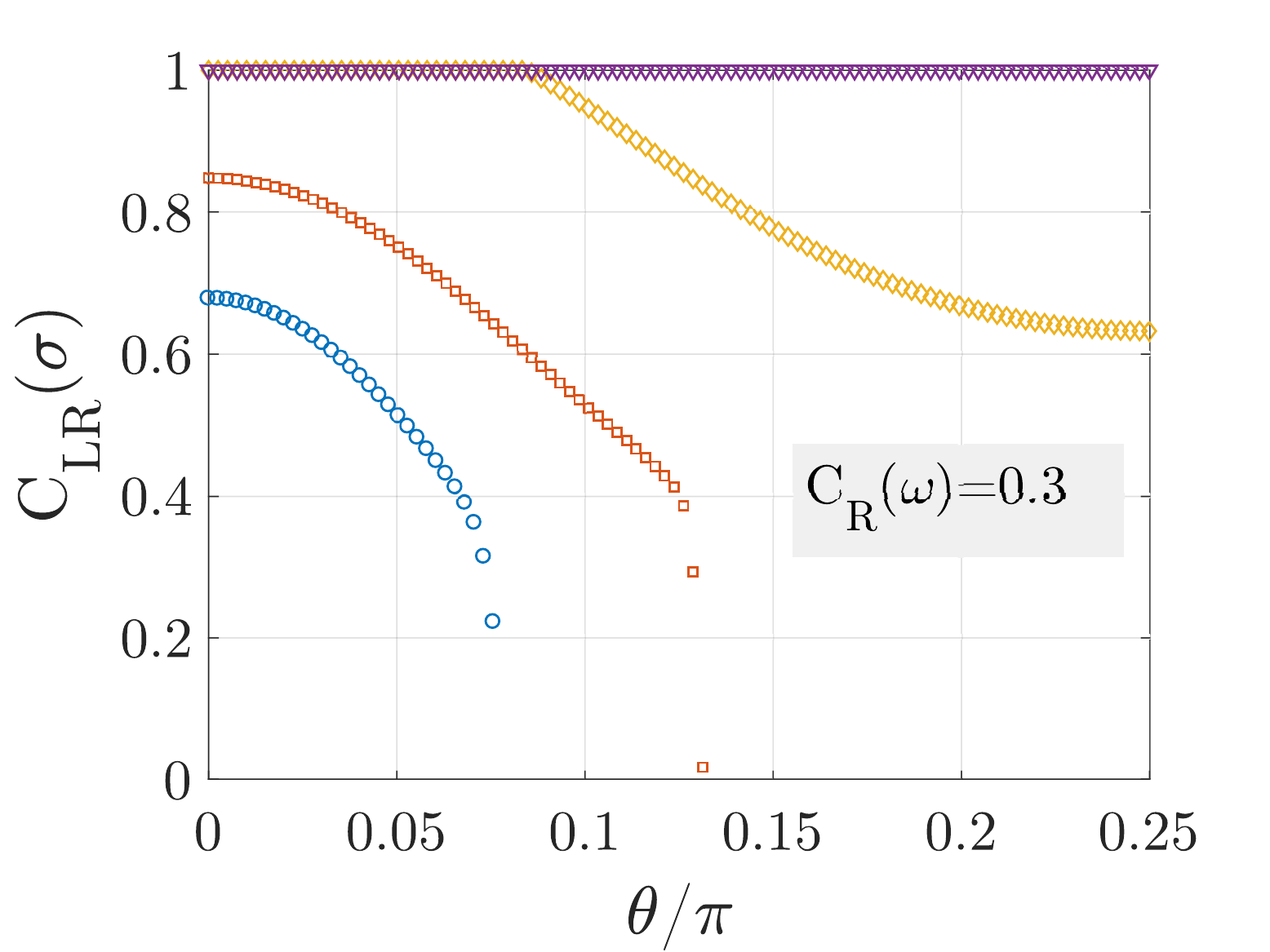}}\\
	
	\subfloat[]{ \includegraphics[width = 0.5\linewidth]{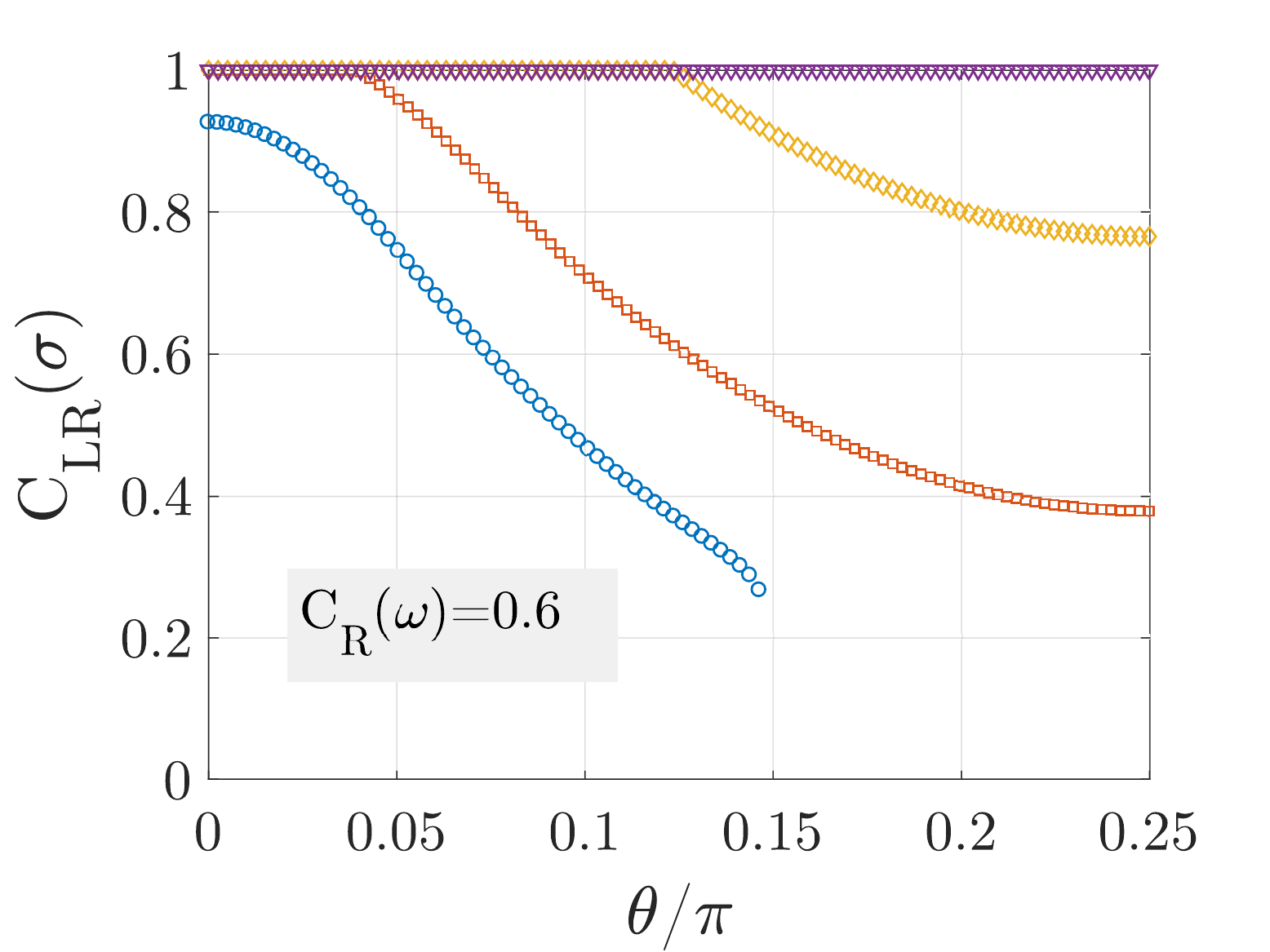}}
	\subfloat[]{ \includegraphics[width = 0.5\linewidth]{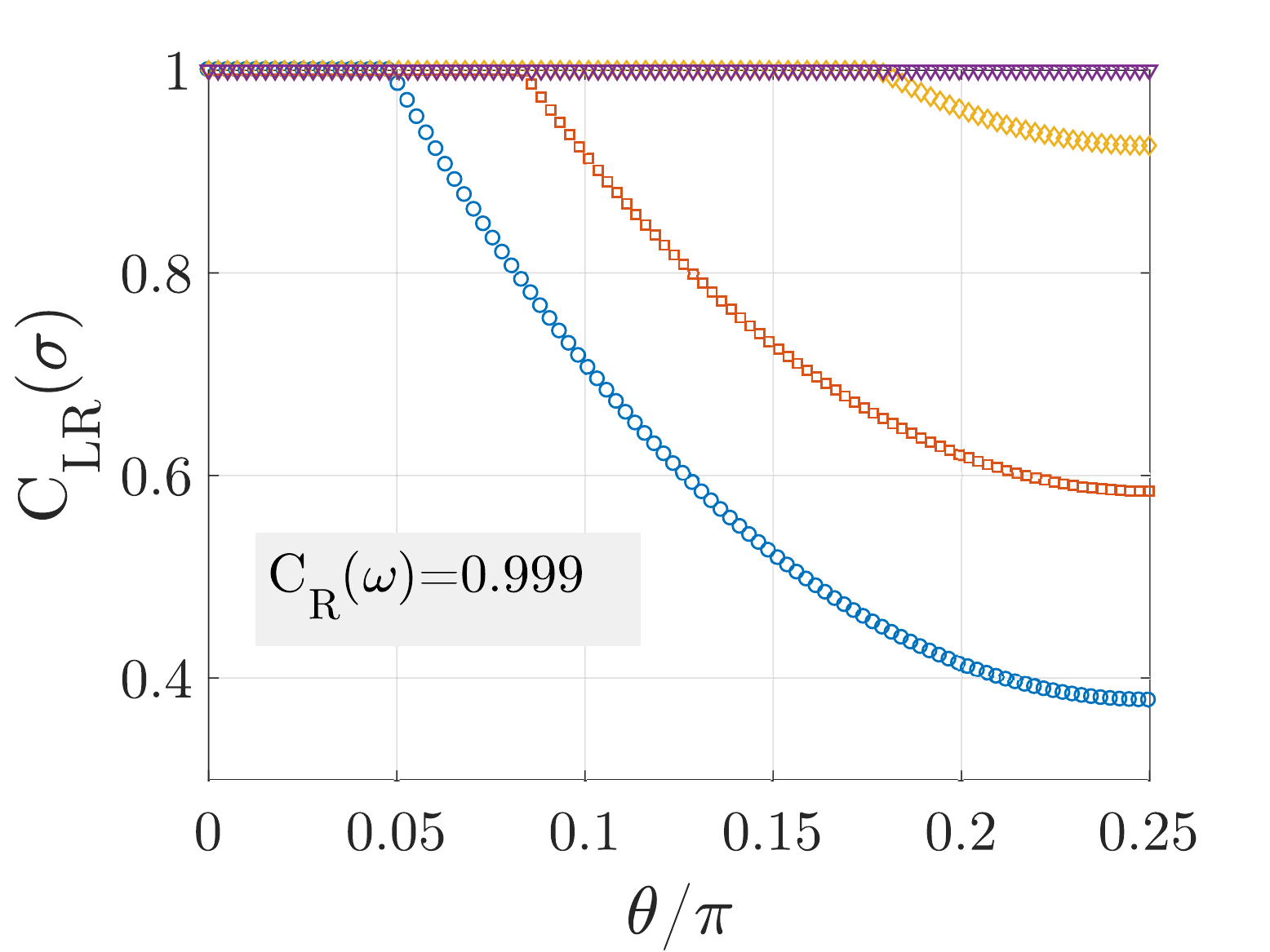}}
 	\caption{
 	         Maximum robustness of coherence left in the resource after the  
 	        approximate MIO implementation of the qubit unitary $U_{\theta}$ vs. $\theta/\pi$ for four values of input coherence (a-d). Each plot shows four curves corresponding to different error thresholds $\epsilon\in\{0.15,0.25,0.45,0.75\}$ (from lower to higher curves).}
 	\label{fig:coh_left_vs_uni}
 \end{figure*}
 %\twocolumngrid

%%%%%%%%%%%%%%%%%%%%%%%%%%%%%%%%%%%%%%%%%%%%%%%%%%%%%%%%%%%%%%%
\subsection{State conversion under MIO and flagpole states}
\label{app:flagpoles}
In this subsection we show that  any pure resource state in dimension larger than $2$ can implement some coherent unitary channel.  
We will first establish a necessary (but not sufficient) criterion for state conversion under MIO (Lemma~\ref{lemma1}), and then 
show that there exist families of flagpole states (see Definition~\ref{def:flagpole}), not convertible to cosdits, that can be used 
as resources for implementing \emph{any} quantum channel.  In particular, we will demonstrate that for any channel 
$\cN: A\to B$, there exists $\cM: R\otimes A\to B,\; 
\cM \text{ MIO}$ and a particular family of pure flagpole states $\ket{\phi_p}\in R$ capable of simulating the channel  
(Theorem~\ref{theorem2}). 

\begin{proof-of}{Lemma \ref{lemma1}}
The geometric measure of coherence, introduced in~\cite{streltsov2015}, is defined as 
\begin{equation}
C_g(\rho)=1-F_C(\rho) 
\label{app:geomteric_measure}
\end{equation}
where 
\begin{equation}\label{fidofcoh}
  F_C(\rho) = \max_{\delta\in\Delta} F(\rho,\delta)^2,
\end{equation}
is the familiar Uhlmann fidelity between two mixed states, $F(\rho,\delta) = \| \sqrt{\rho}\sqrt{\delta} \|_1$.  In particular for a 
pure state $\rho=\dketbra{\psi},\;F(\psi,\delta)^2 = \tr \psi \delta$, and $F_C(\psi) = \lambda_1(\psi) = \max_i \bra{i}\psi\ket{i}$.
As the geometric measure of coherence is monotonically decreasing under MIO, it follows that 
$F_C(\cM(\rho))\geq F_C(\rho)$, for all $\cM$ MIO.  Setting $\phi=\cM(\psi)$ leads to the desired result.
\end{proof-of}

\begin{proof-of}{Theorem \ref{theorem2}}
Let us decompose the state space $ R$ into the following orthogonal subspaces
\begin{equation}
\varphi_p,\;\phi,\; \Pi=\1-\varphi_p-\phi
\label{app:orthogonal_complements}
\end{equation} 
where $\varphi_p=\dketbra{\varphi_p}$, with $\ket{\varphi_p}$ given by Eq.~\eqref{eq:flagpole}, $\phi_1$ is the rank-one 
projector onto the state
\begin{equation}
\ket{\phi}= \sqrt{1-p}\,\ket{0}-\sqrt{\dfrac{p}{d-1}}\sum_{j=1}^{d-1}\ket{j},
\label{app:phi1}
\end{equation}
and $\Pi$ is a rank-$(d-2)$ orthogonal complement to $\varphi_p$ and $\phi$.

We now make the following ansatz for the MIO channel simulating $\cN$ with a flagpole resource $\varphi_p$:
\begin{align}
\cM(\rho\otimes \sigma) 
  &=\cN(\rho)\tr(\sigma \varphi_p)+\cL_1(\rho)\tr(\sigma \phi_1)\nonumber\\
  &\phantom{=}+\cL_2(\rho)\tr(\sigma \Pi),
\label{app:MIO_sim_ansatz}
\end{align}
where $\cL_{1,2}$ are CPTP maps. In order for $\cM$ to be MIO we require that $\cM(\cdot \otimes \ket{j}\!\bra{j})$ is MIO, $\forall\,j\in(0,\ldots,d-1)$. This leads to the following two conditions
\begin{align}\nonumber
  &p \cN + (1-p) \cL_1 \text{ is MIO}&\mathrm{for}\, j=0\\
  &\frac{1-p}{d-1} \cN + \dfrac{p}{d-1} \cL_1 + \left(1-\frac{1}{d-1}\right) \cL_2  \text{ is MIO}&\mathrm{for}\, j>0
\label{app:MIO_requirements}
\end{align}
From Proposition~\ref{app:prop_2} we know that there exists  a CPTP map $\cL_1$ such that the first condition in 
Eq.~\eqref{app:MIO_requirements} is fulfilled if and only if $p\leq(1+C_R(\cN))^{-1}$.   To prove that the second condition in 
Eq.~\eqref{app:MIO_requirements} is satisfied note that it may be written as 
\begin{equation}
\dfrac{1}{d-1}\cN' + \left(1-\dfrac{1}{d-1}\right) \cL_2
\label{app:requirement_2_}
\end{equation}
where we defined the CPTP map $\cN'=(1-p)\cN+p \cL_1$ whose  robustness of coherence is at most $d-1$. By Definition~\ref{app:prop_2} this implies that there exists a CPTP map $\cL_2$ such that $(\cN'+(d-2)\cL_2)/
(d-1)\text{ is MIO}$.  This completes the proof.
\end{proof-of}

\subsection{Qubit unitaries}
In this appendix we  prove Propositions~\ref{proposition_qubit_qubit} and~\ref{proposition_qubit_qutrit} pertaining to the 
implementation of unitary gates via MIO using non-maximally coherent qubit states and $d$-dimensional flagpole states respectively.     
We first prove that exact implementation of unitary gates via MIO can only be achieved using cosbit resources, and then 
establish the subset of $d$-dimensional flagpole states that allow for the implementation of any qubit unitary via MIO.

\begin{proof-of}{Proposition \ref{proposition_qubit_qubit}}
Let $\{\ket{\omega},\,\ket{\omega^\perp}\}$ be an orthonormal basis for $\mathbb{C}^2$, where
\begin{align}\nonumber
\ket{\omega}&=\alpha\ket{0}+\beta \ket{1}\\
\ket{\omega^\perp}&=\beta\ket{0}-\alpha \ket{1}
\label{app:qubit_onb}
\end{align}
with $\alpha,\beta \in \mathbb{C},\, |\alpha|^2+|\beta|^2=1$.  The Kraus decomposition of the most general 
$\cM \text{ is MIO}$ implementing any unitary operator $U:{\cal H}_d\to{\cal H}_d$, is given by 
\begin{equation}
K_i=\lambda_i U\otimes \ket{0}\bra{\omega}+R^i\otimes \ket{0}\bra{\omega^\perp},
\label{app:MIO_Kraus}
\end{equation}
where $\{\lambda_i\in\mathbb{C}\}$, and $\{R^i:{\cal H}_d\to{\cal H}_{d'},\, 1\leq d'\leq d\}$.  As $\cM$ is CPTP, 
$\sum_{i}K_i^\dagger K_i = \mathds{1}$ which yields the following conditions on $\{\lambda_i,\, R^i\}$
\begin{eqnarray}
\sum_i|\lambda_i|^2&=&1  \nonumber\\
\sum_i|\lambda_i|^2 R^i&=&0\nonumber\\
\sum_i |\lambda_i|^2R^{i\dagger}  R^i&=&\mathds{1}.
\label{app:qubit_conditions}
\end{eqnarray}
Applying them together with the MIO condition 
$$\bra{m}\Big(\sum_{i} K_i \ket{jk}\bra{jk}K_i^\dagger \Big) \ket{n}=0,\quad \forall\, m\neq n$$ 
gives rise to the following two equations
\begin{eqnarray}
|\alpha|^2 U_{mj}U^{\dagger}_{jn}  +|\beta|^2\sum_{i}|\lambda_i|^2R^i_{mj}(R^{i\dagger})_{jn}&=&0 \label{1proofqubit}\\
|\beta|^2 U_{mj}U^{\dagger}_{jn} +|\alpha|^2\sum_{i}|\lambda_i|^2R^i_{mj}(R^{i\dagger})_{jn}&=&0 \label{3proofqubit},
\end{eqnarray}	
corresponding to $k=0$ and $k=1$ respectively. Summing both equations above we obtain the additional condition
 $\sum_{i}|\lambda_i|^2R^i_{mj}(R^{i\dagger})_{jn}=-U_{mj}U^{\dagger}_{jn}$
which, upon substituting into any one of Eqs.~(\ref{1proofqubit},~\ref{3proofqubit}) results in  
\begin{equation}
\label{app:condcc}
(|\alpha|^2 -|\beta|^2 )U_{mj}U^{\dagger}_{jn}=0,\quad \forall\, j,n,m.
\end{equation}
Now, unless $U$ is the identity, there is at least one pair of values $(j,n)$ such that $0<|U_{nj}|<1$.  Moreover, as 
$\sum_{m} |U_{mj}|^{2}=1$ there also exists at least one $m\neq n$ such that $|U_{mj}|>0$. Thus for Eq.~\eqref{app:condcc} to 
hold true for all $j,n,m$, $|\alpha|^2=|\beta|^2=\frac{1}{2}$ which implies that  $\ket{\omega}=\ket{\Psi_2}$. This completes the 
proof.
\end{proof-of}

\begin{comment}
{\bf Numerical results for coherence left}
\begin{figure}[H]
		\centering
	\subfloat[]{%
		\includegraphics[width=.7\linewidth]{cRvsthetaC}%
		}\hspace{.02\linewidth}
	\subfloat[]{%
		\includegraphics[width=.7\linewidth]{clvsClthetaB}	
	}
	\caption{
	    a)   Maximal robustness of coherence left in the resource after the optimal MIO simulation  (optimal gate fidelity) for different input qubit resource states $\omega$ as 
a function of the cohering angle $\theta$ of the target unitary $U_{\theta}$.
b)  Maximal log-robustness of coherence left in the resource after the optimal MIO simulation   (optimal gate fidelity) for different input qubit resource states $\omega$ as 
a function of the log-robustness of the target unitary $U_{\theta}$.
In a) and b) different curves correspond $C_{R}(\omega)\in\{0.1, 0.3, 0.6, 0.9, 0.99, 0.999, 1\}$, from lower to higher curves respectively.
 In particular, the last (higher) of  these curves, corresponding to a cosbit input resource, coincides with the theoretical prediction \eqq{upBoundCohLeft} (light blue *). 
 The data are obtained with the SDP of Proposition~\ref{sdp_cohleft}.}
	\label{fig:coh_left_vs_uni}
\end{figure}
\end{comment}

\begin{proof-of}{Proposition \ref{proposition_qubit_qutrit}}
We shall first prove the statement of the Proposition for the case of qutrit flagpoles.  To that end consider the following  
orthonormal basis for $\mathbb{C}^3$,
\begin{align*}
 \ket{\varphi_p} \equiv \ket{\phi_0}&=  \sqrt{p}\ket{0}+\sqrt{\frac{1-p}{2}}(\ket{1}+\ket{2}), \\
  \ket{\phi_1} &=  \frac{1}{\sqrt{2}}(\ket{1}-\ket{2}), \\
  \ket{\phi_2} &= \sqrt{1-p}\ket{0}-\sqrt{\frac{p}{2}}(\ket{1}+\ket{2}).
\end{align*}
As in the previous case, the Kraus decomposition of the most general MIO $\cM$ implementing 
$U_\theta:\mathbb{C}^2\to\mathbb{C}^2$ is given by 
\begin{equation}\label{app:kraus_ops_qutrit}
K_i=\lambda_i U_\theta \otimes \ket{0}\bra{\phi_0}+R^i\otimes \ket{0}\bra{\phi_1}+P^i\otimes \ket{0}\bra{\phi_2}.
\end{equation}
As $\cM$ is CPTP, $\sum_{i}K_i^\dagger K_i =\mathds{1}$, which imposes the following conditions on 
$\{\lambda_i,\, R^i,\, P^i\}$
\begin{eqnarray}
\sum_i|\lambda_i|^2&=&1\\
\sum_i |\lambda_i|^2R^i&=&\sum_\alpha P_\alpha=0\\
\sum_i |\lambda_i|^2 R^{i\dagger}  R^i&=&\sum_i |\lambda_i|^2 P^{i\dagger} P^i=\mathds{1}\label{app:RPareCPTP}\\
\sum_i |\lambda_i|^2 R^{i\dagger}  P^i&=&\sum_i |\lambda_i|^2 P^{i\dagger} R^i=0.	
\end{eqnarray}
Applying the above conditions together with the MIO condition 
$$\bra{m}\left(\sum_{\alpha} K_\alpha \ket{jk}\bra{jk}K_\alpha^\dagger \right) \ket{n}=0$$ 
for $m\neq n=0,1$, $j=0,1$ and $k=0,1,2$ we obtain
\begin{equation}\label{app:mio_qutrit}
U_{mj}U_{jn}^\dagger |\phi_0^k|^2+A_j|\phi_1^k|^2+B_j|\phi_2^k|^2+C_j\phi_1^k\phi_2^k=0,
\end{equation}
where we have defined 
\begin{align*}
A_j&=\sum_i|\lambda_i|^2 R_{mj}^i (R^{i\dagger})_{jn}\\ 
B_j&=\sum_i |\lambda_i|^2 P_{mj}^i (P^{i\dagger})_{jn}\\
C_j&=\sum_i|\lambda_i|^2 (R_{mj}^i (P^{i\dagger})_{jn}+P_{mj}^i (R^{i\dagger})_{jn}),
\end{align*}
and $\phi_l^k\equiv\braket{\phi_l}{k}$.
Fixing $m=0$, $n=1$, Eq.~\eqref{app:mio_qutrit} gives rise to six equations which we can solve for $A_j$, $B_j$ and $C_j$ to 
obtain
\begin{eqnarray}\nonumber
A_0&=&-A_1=-\dfrac{(2p-1)}{p-1}\cos\theta\sin\theta\\ \nonumber
B_0&=&-B_1=\dfrac{p}{p-1}\cos\theta\sin\theta\\
C_0&=&C_1=0.
 \label{app:Aj}
\end{eqnarray}

Now by definition $A_j$ and $B_j$ correspond to the off-diagonal elements of CPTP maps acting on the state 
$\dketbra{j}$.  As the resulting operator must be positive it follows that $|A_j|\leq\frac{1}{2},\,|B_j|\leq\frac{1}{2}$.  Moreover, the 
conditions on the $A_j$ in Eq.~\eqref{app:Aj} are satisfied if $B_0,\, B_1$ are given as in Eq.~\eqref{app:Aj}. Hence  
$$|B_j|=\dfrac{p}{1-p}\cos\theta\sin\theta \leq \dfrac{1}{2},$$
which implies
$$p\leq \dfrac{1}{1+\sin2\theta}.$$
Therefore, there can be no MIO that implements the qubit unitary, $U_\theta$, with a qutrit flagpole state 
$\ket{\varphi_p}$ such that $p>(1+\sin2\theta)^{-1}$.

The proof can be extended to arbitrary $d$-dimensional flagpole and 
$(d-1)$-dimensional unitary.  We note, however, that in this case we obtain the following upper bound 
$p \leq \left({1+\frac{C_R(U)}{d-2}}\right)^{-1}$, 
which coincides with \eqq{pqutrit} only for $d=3$. 
This completes the proof.
\end{proof-of}

We note that it is an interesting open question whether less coherent flagpole states allow for the implementation of qubit 
unitaries via MIO for the case $d>3$.

\subsection{Asymptotic MIO generation capacity}
In this appendix we prove Thm.~\ref{asyCohGen}, namely that the asymptotic coherence-generating capacity of a channel under MIOs is equal to its complete relative-entropy cohering power. The proof is a straightforward extension of \cite[Thm.~1]{bendana2017}, where IOs where considered instead. The key difference here is that the coherence cost of arbitrary states under MIOs is given by the relative entropy of coherence~\cite{zhao2018}, whereas under IOs this is true only for pure states. 

%Let us first recall the generic protocol for coherence generation with many uses of a channel $\cN:A\rightarrow B$: i) an initial incoherent state $\rho_0\in AC$ undergoes the action of the channel, then is post-processed by a MIO $M_1:BC\rightarrow AC_1$, obtaining the state $\rho_1=M_1\circ(N\otimes\id)(\rho_0)$; ii) the same process is iterated $n-1$ times, producing the intermediate states $\rho_j=M_j\circ(N\otimes\id)(\rho_{j-1})$, $j=1,\cdots,n$; iii) the iterations stop when the final state $\rho_n$ is $\epsilon$-close to the cosdit $\Psi_{2}^{\otimes nR}$, i.e., $\tr(\rho_n\Psi_{2}^{\otimes nR})\geq1-\epsilon$, with $R$ being the resulting coherence generating rate. Then the coherence-generating capacity of $\cN$ is defined as follows:
%\begin{equation}
%C_{gen, MIO}^{(\infty)}(\cN):=\sup_{\epsilon>0}\limsup_{n\rightarrow\infty}R.
%\end{equation}

\begin{proof-of}{Theorem \ref{asyCohGen}}
We first prove the upper bound: for any coherence-generating protocol, the trace-distance between the final state $\rho_n$ and the cosdit can be upper bounded as $||\rho_n-\Psi_2^{\otimes nR}||_1\leq 2\sqrt{\epsilon}$. Then the asymptotic continuity of the relative-entropy of coherence~\cite[Lemma 12]{winter2016} implies that:
\begin{equation}\label{asyDist}
\left|C_r(\Psi_2^{\otimes nR})-C_r(\rho_n)\right|\leq 2\sqrt{\epsilon}~nR+2h(\sqrt{\epsilon}),
\end{equation}
where $h(x)=-x\log x-(1-x)\log(1-x)$ is the binary-entropy function. 
Hence we obtain the following chain of inequalities:
\begin{equation}\begin{split}
nR(1-2\sqrt{\epsilon})&\ -2h(\sqrt{\epsilon}) \leq C_r(\rho_n)\\
&\ =\sum_{j=0}^{n-1}\left(C_r(\rho_{j+1})-C_r(\rho_j)\right)\\
&\ \leq n\sup_{\rho\in AC} \left(C_r((\cN\otimes\id)(\rho))-C_r(\rho)\right),
\end{split}
\end{equation}
where the first inequality follows from \eqq{asyDist} and $C_r(\Psi_2^{\otimes nR})=nR$, while the equality by adding and subtracting the relative-entropy of coherence of each intermediate state of the protocol. The last inequality instead follows from substituting each term of the sum with its $\sup$ over all states and the monotonicity of $C_r$ under MIOs. After dividing both sides by $n$ and taking the $n\rightarrow\infty$ limit we obtain that the rate of any coherence-generating protocol is upper-bounded by $\mathbb{P}_r(\cN)$ up to an error vanishing with $\epsilon$. 

For the lower bound instead we need to exhibit a protocol that asymptotically attains $\mathbb{P}_r(\cN)$. The protocol works as follows: i) we first apply the channel to an incoherent state $\dketbra{0}$ in order to produce some coherence, i.e., $\sigma=\cN(\dketbra{0})$; ii) we produce a sufficient number of copies of $\sigma$ to distill a certain amount of cosbits, which are then used to produce a target state $\rho$ via coherence dilution; iii) we apply the channel to $\rho$ in order to obtain a more coherent state $\rho'=(\cN\otimes\id)(\rho)$; iv) we distill cosbits from $\rho'$; v) we use the increased amount of cosbits obtained to iterate the processes (iii-iv) $k$ times. 
The asymptotic coherence-generating capacity and cost of states under MIOs are both equal to the relative-entropy of coherence~\cite{zhao2018}, so that the conversion rates of processes (ii-iv) described above can be written as
\begin{equation}\begin{split}\label{convRates}
\sigma^{\otimes m} &\ \mapsto\Psi_2^{\otimes\lfloor m(C_r(\sigma)-\delta)\rfloor},\\
\Psi_2^{\otimes\lceil n(C_r(\rho)+\delta)\rceil}&\ \mapsto\rho^{\otimes n}\\
\rho'^{\otimes n}&\ \mapsto\Psi_2^{\otimes\lfloor n(C_r(\rho')-\delta)\rfloor},
\end{split}\end{equation} 
where $m\gtrsim n (C_r(\rho)+\delta)/(C_r(\sigma)-\delta)$.
Note that the coherence production of processes (iii-iv) is not larger than $n(R-2\delta)$, where $R=C_r(\rho')-C_r(\rho)$, and each transformation is accurate up to an error $\epsilon$.
Hence, the overall coherence-production rate is bounded by 
\begin{equation}
\frac{kn(R-2\delta)+n(C_r(\phi)+\delta)}{kn+m}\xrightarrow[k\rightarrow\infty]{}R-2\delta,
\end{equation}
which can be made arbitrarily close to $R$ by taking $\delta$ sufficiently small. We have therefore described a protocol that generates coherence at an asymptotical rate $C_r((N\otimes\id)(\phi))-C_r(\phi)$ for any $\phi_{AC}$ using MIOs. By taking the $\sup$ of this quantity, we obtain the desired lower bound. 
\end{proof-of}

\end{document}